\documentclass[journal]{IEEEtran} %
\usepackage{cite}
\usepackage[bookmarks=false]{hyperref}
\usepackage{subfig}
\usepackage[utf8]{inputenc} 
\usepackage{graphicx,url}
\usepackage{booktabs}       
\graphicspath{ {./images/} }
\usepackage{multirow}
\usepackage{hhline}
\usepackage{graphicx}
\usepackage{amssymb,amsmath,mathtools}
\usepackage{amsthm}
\newtheorem{theorem}{Theorem}

\newtheorem{lemma}{Lemma}

\usepackage{bm}                 %
\usepackage{xspace}

\usepackage{color,soul}
\usepackage[dvipsnames]{xcolor}
\hypersetup{
    colorlinks,
    linkcolor={blue!50!black},
    citecolor={blue!50!black},
    urlcolor={blue!80!black}
}
\usepackage{bbm}
\usepackage{lipsum,lineno}
\usepackage{float}
\makeatletter
\newif\if@restonecol
\makeatother
\usepackage{amsmath}
\usepackage{mathrsfs}
\usepackage[shortlabels]{enumitem}
\usepackage{algorithm}
\usepackage{algpseudocode}
\usepackage{fancyhdr} 
\usepackage{tabularx}
\usepackage{dsfont}
\usepackage{comment}               
\makeatletter
\renewenvironment{proof}[1][\proofname]{\par
  \pushQED{\qed}%
  \normalfont \topsep6\p@\@plus6\p@\relax
  \trivlist
  \item[\hskip\labelsep
        \itshape
    #1]\ignorespaces
}{%
  \popQED\endtrivlist\@endpefalse
}
\makeatother
\usepackage{tikz,pgfplots}
\usepackage{pgfplots}
\usetikzlibrary{shapes.geometric,backgrounds,patterns, trees}
\usetikzlibrary{3d,decorations.text,shapes.arrows,positioning,fit,backgrounds}
\usetikzlibrary{positioning, decorations.pathmorphing, shapes}
\usetikzlibrary{decorations.pathreplacing}
\usetikzlibrary{shapes.geometric,backgrounds,patterns, trees}
\usetikzlibrary{arrows.meta,
                bending,
                intersections,
                quotes,
                shapes.geometric}
              \usetikzlibrary{automata, positioning}
\usepgfplotslibrary{fillbetween}
\usetikzlibrary{shapes,arrows}
\usetikzlibrary{arrows.meta}
\usetikzlibrary{positioning}
\tikzset{set/.style={draw,circle,inner sep=0pt,align=center}}
\usetikzlibrary{automata, positioning}
  \usetikzlibrary{shapes,shadows}
  \tikzstyle{abstractbox} = [draw=black, fill=white, rectangle,
  inner sep=10pt, style=rounded corners, drop shadow={fill=black,
  opacity=1}]
\tikzstyle{abstracttitle} =[fill=white]
\usetikzlibrary{calc,positioning,shapes.geometric}
\usetikzlibrary{arrows.meta,arrows}
\DeclareMathOperator*{\argmax}{arg\,max}
\DeclareMathOperator*{\argmin}{arg\,min}

\usetikzlibrary{matrix}
\tikzstyle{cblue}=[circle, draw, thin,fill=cyan!20, scale=0.8]
\tikzstyle{qgre}=[rectangle, draw, thin,fill=green!20, scale=0.8]
\tikzstyle{rpath}=[ultra thick, red, opacity=0.4]
\tikzstyle{legend_isps}=[rectangle, rounded corners, thin,
                       fill=gray!20, text=blue, draw]

\tikzstyle{legend_overlay}=[rectangle, rounded corners, thin,
                           top color= white,bottom color=green!25,
                           minimum width=2.5cm, minimum height=0.8cm,
                           pinegreen]
\tikzstyle{legend_phytop}=[rectangle, rounded corners, thin,
                          top color= white,bottom color=cyan!25,
                          minimum width=2.5cm, minimum height=0.8cm,
                          royalblue]
\tikzstyle{legend_general}=[rectangle, rounded corners, thin,
                          top color= white,bottom color=lavander!25,
                          minimum width=2.5cm, minimum height=0.8cm,
                          violet]
                          \usetikzlibrary{matrix}
                          
\allowdisplaybreaks
\colorlet{myRed}{red!20}
\tikzset{
  rows/.style 2 args={/utils/temp/.style={row ##1/.append style={nodes={#2}}},
    /utils/temp/.list={#1}},
  columns/.style 2 args={/utils/temp/.style={column ##1/.append style={nodes={#2}}},
    /utils/temp/.list={#1}}}
\usetikzlibrary{backgrounds,calc,shadings,shapes.arrows,shapes.symbols,shadows}
\definecolor{switch}{HTML}{006996}

\makeatletter
\pgfkeys{/pgf/.cd,
  parallelepiped offset x/.initial=2mm,
  parallelepiped offset y/.initial=2mm
}
\pgfdeclareshape{parallelepiped}
{
  \inheritsavedanchors[from=rectangle] 
  \inheritanchorborder[from=rectangle]
  \inheritanchor[from=rectangle]{north}
  \inheritanchor[from=rectangle]{north west}
  \inheritanchor[from=rectangle]{north east}
  \inheritanchor[from=rectangle]{center}
  \inheritanchor[from=rectangle]{west}
  \inheritanchor[from=rectangle]{east}
  \inheritanchor[from=rectangle]{mid}
  \inheritanchor[from=rectangle]{mid west}
  \inheritanchor[from=rectangle]{mid east}
  \inheritanchor[from=rectangle]{base}
  \inheritanchor[from=rectangle]{base west}
  \inheritanchor[from=rectangle]{base east}
  \inheritanchor[from=rectangle]{south}
  \inheritanchor[from=rectangle]{south west}
  \inheritanchor[from=rectangle]{south east}
  \backgroundpath{
    \southwest \pgf@xa=\pgf@x \pgf@ya=\pgf@y
    \northeast \pgf@xb=\pgf@x \pgf@yb=\pgf@y
    \pgfmathsetlength\pgfutil@tempdima{\pgfkeysvalueof{/pgf/parallelepiped
      offset x}}
    \pgfmathsetlength\pgfutil@tempdimb{\pgfkeysvalueof{/pgf/parallelepiped
      offset y}}
    \def\ppd@offset{\pgfpoint{\pgfutil@tempdima}{\pgfutil@tempdimb}}
    \pgfpathmoveto{\pgfqpoint{\pgf@xa}{\pgf@ya}}
    \pgfpathlineto{\pgfqpoint{\pgf@xb}{\pgf@ya}}
    \pgfpathlineto{\pgfqpoint{\pgf@xb}{\pgf@yb}}
    \pgfpathlineto{\pgfqpoint{\pgf@xa}{\pgf@yb}}
    \pgfpathclose
    \pgfpathmoveto{\pgfqpoint{\pgf@xb}{\pgf@ya}}
    \pgfpathlineto{\pgfpointadd{\pgfpoint{\pgf@xb}{\pgf@ya}}{\ppd@offset}}
    \pgfpathlineto{\pgfpointadd{\pgfpoint{\pgf@xb}{\pgf@yb}}{\ppd@offset}}
    \pgfpathlineto{\pgfpointadd{\pgfpoint{\pgf@xa}{\pgf@yb}}{\ppd@offset}}
    \pgfpathlineto{\pgfqpoint{\pgf@xa}{\pgf@yb}}
    \pgfpathmoveto{\pgfqpoint{\pgf@xb}{\pgf@yb}}
    \pgfpathlineto{\pgfpointadd{\pgfpoint{\pgf@xb}{\pgf@yb}}{\ppd@offset}}
  }
}

\makeatletter
\tikzset{anchor/.append code=\let\tikz@auto@anchor\relax,
  add font/.code=%
    \expandafter\def\expandafter\tikz@textfont\expandafter{\tikz@textfont#1},
  left delimiter/.style 2 args={append after command={\tikz@delimiter{south east}
    {south west}{every delimiter,every left delimiter,#2}{south}{north}{#1}{.}{\pgf@y}}}}
\tikzstyle{sms} = [rectangle callout, draw,very thick, rounded corners, minimum height=20pt]
\makeatletter
\tikzset{anchor/.append code=\let\tikz@auto@anchor\relax,
  add font/.code=%
    \expandafter\def\expandafter\tikz@textfont\expandafter{\tikz@textfont#1},
  left delimiter/.style 2 args={append after command={\tikz@delimiter{south east}
    {south west}{every delimiter,every left delimiter,#2}{south}{north}{#1}{.}{\pgf@y}}}}
\tikzstyle{sms} = [rectangle callout, draw,very thick, rounded corners, minimum height=20pt]
\usetikzlibrary{positioning,calc}

\tikzset{l3 switch/.style={
    parallelepiped,fill=switch, draw=white,
    minimum width=0.75cm,
    minimum height=0.75cm,
    parallelepiped offset x=1.75mm,
    parallelepiped offset y=1.25mm,
    path picture={
      \node[fill=white,
        circle,
        minimum size=6pt,
        inner sep=0pt,
        append after command={
          \pgfextra{
            \foreach \angle in {0,45,...,360}
            \draw[-latex,fill=white] (\tikzlastnode.\angle)--++(\angle:2.25mm);
          }
        }
      ]
       at ([xshift=-0.75mm,yshift=-0.5mm]path picture bounding box.center){};
    }
  },
  ports/.style={
    line width=0.3pt,
    top color=gray!20,
    bottom color=gray!80
  },
  rack switch/.style={
    parallelepiped,fill=white, draw,
    minimum width=1.25cm,
    minimum height=0.25cm,
    parallelepiped offset x=2mm,
    parallelepiped offset y=1.25mm,
    xscale=-1,
    path picture={
      \draw[top color=gray!5,bottom color=gray!40]
      (path picture bounding box.south west) rectangle
      (path picture bounding box.north east);
      \coordinate (A-west) at ([xshift=-0.2cm]path picture bounding box.west);
      \coordinate (A-center) at ($(path picture bounding box.center)!0!(path
        picture bounding box.south)$);
      \foreach \x in {0.275,0.525,0.775}{
        \draw[ports]([yshift=-0.05cm]$(A-west)!\x!(A-center)$)
          rectangle +(0.1,0.05);
        \draw[ports]([yshift=-0.125cm]$(A-west)!\x!(A-center)$)
          rectangle +(0.1,0.05);
       }
      \coordinate (A-east) at (path picture bounding box.east);
      \foreach \x in {0.085,0.21,0.335,0.455,0.635,0.755,0.875,1}{
        \draw[ports]([yshift=-0.1125cm]$(A-east)!\x!(A-center)$)
          rectangle +(0.05,0.1);
      }
    }
  },
  server/.style={
    parallelepiped,
    fill=white, draw,
    minimum width=0.35cm,
    minimum height=0.75cm,
    parallelepiped offset x=3mm,
    parallelepiped offset y=2mm,
    xscale=-1,
    path picture={
      \draw[top color=gray!5,bottom color=gray!40]
      (path picture bounding box.south west) rectangle
      (path picture bounding box.north east);
      \coordinate (A-center) at ($(path picture bounding box.center)!0!(path
        picture bounding box.south)$);
      \coordinate (A-west) at ([xshift=-0.575cm]path picture bounding box.west);
      \draw[ports]([yshift=0.1cm]$(A-west)!0!(A-center)$)
        rectangle +(0.2,0.065);
      \draw[ports]([yshift=0.01cm]$(A-west)!0.085!(A-center)$)
        rectangle +(0.15,0.05);
      \fill[black]([yshift=-0.35cm]$(A-west)!-0.1!(A-center)$)
        rectangle +(0.235,0.0175);
      \fill[black]([yshift=-0.385cm]$(A-west)!-0.1!(A-center)$)
        rectangle +(0.235,0.0175);
      \fill[black]([yshift=-0.42cm]$(A-west)!-0.1!(A-center)$)
        rectangle +(0.235,0.0175);
    }
  },
}
\pgfplotsset{compat=1.16}
\usetikzlibrary{calc, shadings, shadows, shapes.arrows}

\tikzset{%
  interface/.style={draw, rectangle, rounded corners, font=\LARGE\sffamily},
  ethernet/.style={interface, fill=yellow!50},
  serial/.style={interface, fill=green!70},
  speed/.style={sloped, anchor=south, font=\large\sffamily},
  route/.style={draw, shape=single arrow, single arrow head extend=4mm,
    minimum height=1.7cm, minimum width=3mm, white, fill=switch!20,
    drop shadow={opacity=.8, fill=switch}, font=\tiny}
}
\newcommand*{\shift}{1.3cm}
\newcommand{\Crossk}{$\mathbin{\tikz [x=1.2ex,y=1.2ex,line width=.1ex, black] \draw (0,0) -- (1,1) (0,1) -- (1,0);}$}%
\newcommand*{\router}[1]{
\begin{tikzpicture}
  \coordinate (ll) at (-3,0.5);
  \coordinate (lr) at (3,0.5);
  \coordinate (ul) at (-3,2);
  \coordinate (ur) at (3,2);
  \shade [shading angle=90, left color=switch, right color=white] (ll)
    arc (-180:-60:3cm and .75cm) -- +(0,1.5) arc (-60:-180:3cm and .75cm)
    -- cycle;
  \shade [shading angle=270, right color=switch, left color=white!50] (lr)
    arc (0:-60:3cm and .75cm) -- +(0,1.5) arc (-60:0:3cm and .75cm) -- cycle;
  \draw [thick] (ll) arc (-180:0:3cm and .75cm)
    -- (ur) arc (0:-180:3cm and .75cm) -- cycle;
  \draw [thick, shade, upper left=switch, lower left=switch,
    upper right=switch, lower right=white] (ul)
    arc (-180:180:3cm and .75cm);
  \node at (0,0.5){\color{blue!60!black}\Huge #1};
  \begin{scope}[yshift=2cm, yscale=0.28, transform shape]
    \node[route, rotate=45, xshift=\shift] {\strut};
    \node[route, rotate=-45, xshift=-\shift] {\strut};
    \node[route, rotate=-135, xshift=\shift] {\strut};
    \node[route, rotate=135, xshift=-\shift] {\strut};
  \end{scope}
\end{tikzpicture}}

\makeatletter
\pgfdeclareradialshading[tikz@ball]{cloud}{\pgfpoint{-0.275cm}{0.4cm}}{%
  color(0cm)=(tikz@ball!75!white);
  color(0.1cm)=(tikz@ball!85!white);
  color(0.2cm)=(tikz@ball!95!white);
  color(0.7cm)=(tikz@ball!89!black);
  color(1cm)=(tikz@ball!75!black)
}
\tikzoption{cloud color}{\pgfutil@colorlet{tikz@ball}{#1}%
  \def\tikz@shading{cloud}\tikz@addmode{\tikz@mode@shadetrue}}
\makeatother

\tikzset{my cloud/.style={
     cloud, draw, aspect=2,
     cloud color={gray!5!white}
  }
}

\renewcommand\thetable{\arabic{table}}
\newcommand{\acro}[1]{\textsc{#1}\xspace}
\newcommand{\acros}[1]{\textsc{#1}s\xspace}
\newcommand{\acrop}[1]{\textsc{#1}\xspace}
\newcommand{\tpp}{\acrop{tp-2}}
\newcommand{\posg}{\acro{posg}}
\newcommand{\posgs}{\acros{posg}}

\newcommand{\mdp}{\acro{mdp}}
\newcommand{\pomdp}{\acro{pomdp}}

\newcommand{\mdps}{\acros{mdp}}
\newcommand{\pomdps}{\acros{pomdp}}
\newcommand{\idps}{\acro{idps}}

\newcommand{\ipss}{\acros{ips}}

\newcommand{\idss}{\acros{ids}}

\newcommand{\irss}{\acros{irs}}

\newcommand{\kl}{\acro{kl}}

\newcommand{\cpu}{\acro{cpu}}
\newcommand{\cgroups}{\acro{cgroups}}

\newcommand{\alphago}{\acro{alphago}}
\newcommand{\ssh}{\acro{ssh}}
\newcommand{\openai}{\acro{openai}}
\newcommand{\five}{\acro{five}}
\newcommand{\irc}{\acro{irc}}
\newcommand{\smtp}{\acro{smtp}}
\newcommand{\snmp}{\acro{snmp}}
\newcommand{\icmp}{\acro{icmp}}
\newcommand{\ntp}{\acro{ntp}}

\newcommand{\debian}{\acro{debian}}
\newcommand{\jessie}{\acro{jessie}}
\newcommand{\wheezy}{\acro{wheezy}}
\newcommand{\samba}{\acro{samba}}

\newcommand{\elastic}{\acro{elastic}}
\newcommand{\tomcat}{\acro{tomcat}}

\newcommand{\dns}{\acro{dns}}
\newcommand{\snort}{\acro{snort}}
\newcommand{\http}{\acro{http}}

\newcommand{\netem}{\acro{netem}}
\newcommand{\apache}{\acro{apache}}
\newcommand{\ts}{\acro{ts}}
\newcommand{\mysql}{\acro{mysql}}
\newcommand{\tcpp}{\acro{tcp}}
\newcommand{\xmas}{\acro{xmas}}
\newcommand{\finn}{\acro{fin}}
\newcommand{\nulll}{\acro{null}}
\newcommand{\udp}{\acro{udp}}
\newcommand{\syn}{\acro{syn}}
\newcommand{\mongo}{\acro{mongodb}}
\newcommand{\postgres}{\acro{postgres}}
\newcommand{\telnet}{\acro{telnet}}
\newcommand{\cassandra}{\acro{cassandra}}
\newcommand{\ubuntu}{\acro{ubuntu}}
\newcommand{\flask}{\acro{flask}}
\newcommand{\icml}{\acro{icml}}
\newcommand{\flipit}{\acro{flipit}}
\newcommand{\mininet}{\acro{mininet}}
\newcommand{\rll}{\acro{rl}}
\newcommand{\sll}{\acro{sl}}
\newcommand{\vulscan}{\acro{vulscan}}

\newcommand{\ftp}{\acro{ftp}}
\newcommand{\proftp}{\acro{proftpd}}

\newcommand{\cve}{\acro{cve}}
\newcommand{\cwe}{\acro{cwe}}

\newcommand{\ppo}{\acro{ppo}}
\newcommand{\sarsa}{\acro{sarsa}}
\newcommand{\hsvi}{\acro{hsvi}}

\newcommand{\spsa}{\acro{spsa}}
\newcommand{\dqn}{\acro{dqn}}
\newcommand{\ddqn}{\acro{ddqn}}
\newcommand{\tfp}{\acro{t-fp}}
\newcommand{\nfsp}{\acro{nfsp}}
\newcommand{\ddpg}{\acro{ddpg}}

\newcommand{\sgs}{\acros{sg}}
\newcommand{\ac}{\acro{a2c}}
\newcommand{\acc}{\acro{a3c}}
\newcommand{\muzero}{\acro{muzero}}
\newcommand{\nfq}{\acro{nfq}}

\usepackage{etoolbox}
\makeatletter
\patchcmd{\@makecaption}
  {\scshape}
  {}
  {}
  {}
\makeatother

\begin{document}
\bstctlcite{MyBSTcontrol}
\title{Learning Near-Optimal Intrusion Responses \\ Against Dynamic Attackers}

\author{\IEEEauthorblockN{Kim Hammar \IEEEauthorrefmark{2}\IEEEauthorrefmark{3} and Rolf Stadler\IEEEauthorrefmark{2}\IEEEauthorrefmark{3}}

 \IEEEauthorblockA{\IEEEauthorrefmark{2}
Division of Network and Systems Engineering, KTH Royal Institute of Technology, Sweden
 }\\
 \IEEEauthorblockA{\IEEEauthorrefmark{3} KTH Center for Cyber Defense and Information Security, Sweden \\
Email: \{kimham, stadler\}@kth.se%
\\
\today
}
}

\maketitle
\begin{abstract}
We study automated intrusion response and formulate the interaction between an attacker and a defender as an optimal stopping game where attack and defense strategies evolve through reinforcement learning and self-play. The game-theoretic modeling enables us to find defender strategies that are effective against a dynamic attacker, i.e. an attacker that adapts its strategy in response to the defender strategy. Further, the optimal stopping formulation allows us to prove that optimal strategies have threshold properties. To obtain near-optimal defender strategies, we develop Threshold Fictitious Self-Play (\tfp), a fictitious self-play algorithm that learns Nash equilibria through stochastic approximation. We show that \tfp outperforms a state-of-the-art algorithm for our use case. The experimental part of this investigation includes two systems: a simulation system where defender strategies are incrementally learned and an emulation system where statistics are collected that drive simulation runs and where learned strategies are evaluated. We argue that this approach can produce effective defender strategies for a practical IT infrastructure.
\end{abstract}


\IEEEpeerreviewmaketitle

\section{Introduction}
An organization's security strategy has traditionally been defined, implemented, and updated by domain experts \cite{int_prevention}. This approach can provide basic security for an organization's communication and computing infrastructure. As infrastructure update cycles become shorter and attacks increase in sophistication, meeting the security requirements becomes increasingly difficult. To address this challenge, significant efforts have started to automate the process of obtaining security strategies \cite{ml_for_cognitive_net_management}. Examples of this research include: computation of defender strategies using dynamic programming and control theory \cite{dp_security_1,Miehling_control_theoretic_approaches_summary}; computation of exploits and corresponding defenses through evolutionary methods \cite{armsrace_malware,hemberg_oreily_evo}; computation of defender strategies through game-theoretic methods \cite{nework_security_alpcan, serkan_gyorgy_game}; derivation of defender responses through causal inference \cite{causal_neil_agent}; use of machine learning techniques to estimate model parameters and strategies \cite{hammar_stadler_tnsm, hammar_stadler_cnsm_20, hammar_stadler_cnsm_21}; automated creation of threat models \cite{mal_pontus}; and identification of infrastructure vulnerabilities through attack simulations and threat intelligence \cite{wagner_automated_segmentation, threat_intel_misp}.

A promising new direction of research is automatically learning security strategies through reinforcement learning methods \cite{rl_bible}, whereby the problem of finding security strategies is modeled as a Markov decision problem and strategies are learned through simulation (see surveys \cite{deep_rl_cyber_sec,control_rl_reviews}). While encouraging results have been obtained following this approach \cite{hammar_stadler_tnsm,hammar_stadler_cnsm_20,hammar_stadler_cnsm_21, hammar_stadler_cnsm_22, elderman, schwartz_2020, oslo_pentest_rl, kurt_rl, microsoft_red_teaming, ridley_ml_defense, rl_cyberdefense_heartbleed, deep_hierarchical_rl_pentest, pentest_rl_rohit, adaptive_cyber_defense_pomdp_rl, muzero_sdn,atmos,sdn_rl_ddos,deep_air,noms_demo_preprint,game_cyber_rl_sim,al_shaer_book_ppo_simulation,cmu_ppo_selfplay,YUNGAICELANAULA2022103444,sdn_zolo,9923774,9916301, 9893186,rl_qcd_kalle,MAEDA2021102108,BLAND2020101738, HUANG2022102844,LIU2021102480,9096506,Khoury2020AHG,Wang2020AnID, han_yi_sdn,stochastic_game_approach_2,Iannucci2021AnIR,Iannucci2019APE,beyond_cage,apt_rl_simulation,zhang2019}, key challenges remain \cite{real_world_rl}. Chief among them is narrowing the gap between the environment where strategies are evaluated and a scenario playing out in a real system. Most of the results obtained so far are limited to simulation environments, and it is not clear how they generalize to practical IT infrastructures. Another challenge is to obtain security strategies that are effective against a \textit{dynamic} attacker, i.e. an attacker that adapts its strategy in response to the defender strategy. Most of the prior work have used reinforcement learning to find effective defender strategies against static attackers, and little is known about the found strategies' performance against a dynamic attacker.
\begin{figure}
  \centering
  \scalebox{0.93}{
    \input{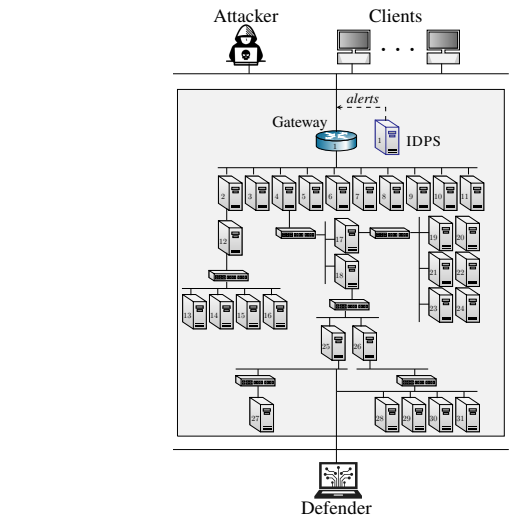}
    }
    \caption{The IT infrastructure and the actors in the intrusion response use case.}
    \label{fig:system2}
  \end{figure}

In this paper, we address the above challenges and present a novel framework to automatically learn a defender strategy against a dynamic attacker. We apply this framework to an \textit{intrusion response} use case, which involves the IT infrastructure of an organization (see Fig. \ref{fig:system2}). The operator of this infrastructure, which we call the defender, takes measures to protect it against an attacker while providing services to a client population.

We formulate the intrusion response use case as an \textit{optimal stopping game}, namely a stochastic game where both players face an optimal stopping problem \cite{wald,dynkin_orig_3,shirayev_change_point}. This formulation enables us to gain insight into the structure of optimal strategies, which we prove to have threshold properties. To obtain effective defender strategies, we use reinforcement learning and self-play. Based on the threshold properties, we design Threshold Fictitious Self-Play (\tfp), an efficient algorithm that iteratively computes near-optimal defender strategies against a dynamic attacker.

Our method for learning and evaluating strategies for a given infrastructure includes building two systems (see Fig. \ref{fig:method}). First, we develop an \textit{emulation system} where key functional components of the target infrastructure are replicated. This system closely approximates the functionality of the target infrastructure and is used to run attack scenarios and defender responses. Such runs produce system measurements and logs, from which we estimate infrastructure statistics, which then are used to instantiate the simulation model.

Second, we build a \textit{simulation system} where game episodes are simulated and strategies are incrementally learned through self-play. Learned strategies are extracted from the simulation system and evaluated in the emulation system.

Two benefits of this method are: (\textit{i}) that the emulation system allows evaluating strategies without affecting operational workflows on the target infrastructure; and (\textit{ii}) that the simulation system enables efficient and rapid learning of strategies. (A video demonstration of the software framework that implements the emulation and simulation systems is available at \cite{fig2_video_demonstration}.)
\begin{figure}
  \centering
  \scalebox{0.95}{
    \input{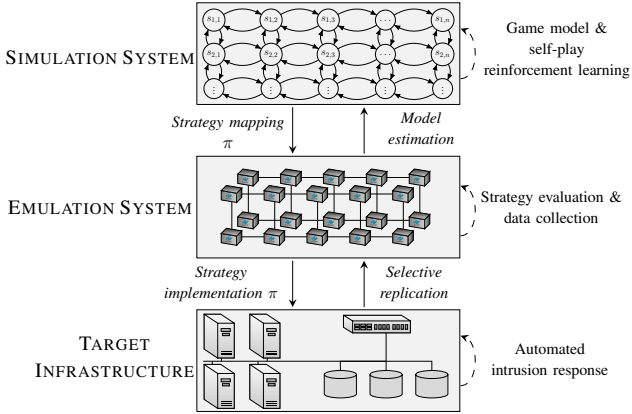}
 }
    \caption{Our framework for finding and evaluating intrusion response strategies \cite{hammar_stadler_tnsm}.}
    \label{fig:method}
\end{figure}

We make three contributions with this paper. First, we formulate intrusion response as an optimal stopping game between an attacker and a defender. This novel formulation allows us a) to derive and prove structural properties of optimal strategies; and b) to find defender strategies that are effective against an attacker with a dynamic strategy. We thus address a key limitation of many related works, which only consider static attackers \cite{hammar_stadler_tnsm,oslo_pentest_rl,ridley_ml_defense,deep_hierarchical_rl_pentest,pentest_rl_rohit,deep_air,hammar_stadler_cnsm_21,hammar_stadler_cnsm_22,adaptive_cyber_defense_pomdp_rl,kurt_rl,YUNGAICELANAULA2022103444,sdn_zolo,9916301, 9893186,7568529,Kreidl2004FeedbackCA,7573127,miehling_attack_graph,MAEDA2021102108,Iannucci2019APE,Iannucci2021AnIR,Wang2020AnID,3560830.3563732,beyond_cage,9984930,model_free_rl_irs}. Second, we propose \tfp, an efficient reinforcement learning algorithm that exploits threshold properties of optimal stopping strategies and outperforms a state-of-the-art algorithm for our use case. Third, we provide evaluation results from an emulated infrastructure. This addresses a drawback in related research that relies solely on simulations to learn and evaluate strategies \cite{hammar_stadler_cnsm_20,hammar_stadler_cnsm_21, elderman, schwartz_2020, oslo_pentest_rl, kurt_rl, microsoft_red_teaming, ridley_ml_defense, rl_cyberdefense_heartbleed, deep_hierarchical_rl_pentest, pentest_rl_rohit, adaptive_cyber_defense_pomdp_rl,serkan_gyorgy_game,flipit,dynamic_game_linan_zhu,general_sum_markov_games_for_strategic_detection_of_apt,game_cyber_rl_sim,cmu_ppo_selfplay,al_shaer_book_ppo_simulation,mec_game_rl_q_learning_sim,nfsp_jamming_1_sim,9923774,9916301,miehling_attack_graph,7573127,7568529,BLAND2020101738,HUANG2022102844,LIU2021102480,HUANG2020101660,9096506,stochastic_game_approach_2,Iannucci2019APE,Wang2021GameTheoreticAI,3560830.3563732,9984930,apt_rl_simulation,zhang2019,fog_computing_irs, altman_jamming_1}.

We believe that this paper provides a foundation for the next generation of security systems, including Intrusion Prevention Systems (\ipss) (e.g. Trellix \cite{trellix}), Intrusion Response Systems (\irss) (e.g. Wazuh \cite{wazuh}), and Intrusion Detection Systems (\idss) (e.g. Snort \cite{snort}). The optimal stopping strategies computed through our framework can be used in these systems to decide at which point in time an automated response action should be triggered or at which point in time a human operator should be alerted to take action.

The work in this paper builds on our earlier results in automated intrusion response \cite{hammar_stadler_cnsm_21, hammar_stadler_tnsm, hammar2022learning}. Specifically, this paper can be seen as a generalization of the work in \cite{hammar_stadler_tnsm}, where we investigate intrusion response against a static attacker. As explained in this paper, intrusion response against a dynamic attacker requires a different mathematical framework. An extended abstract of this paper was presented at the ``Machine learning for cyber security'' workshop at the International Conference on Machine Learning (\icml) 2022 \cite{hammar2022learning}.
\section{The Intrusion Response Use Case}\label{sec:use_case}
We consider an intrusion response use case that involves the IT infrastructure of an organization. The operator of this infrastructure, which we call the defender, takes measures to protect it against an attacker while providing services to a client population (Fig. \ref{fig:system2}). The infrastructure includes a set of servers that run the services and an Intrusion Detection and Prevention System (\idps) that logs events in real-time. Clients access the services through a public gateway, which is also open to the attacker.

The attacker's goal is to intrude on the infrastructure and compromise its servers. To achieve this, the attacker explores the infrastructure through reconnaissance and exploits vulnerabilities while avoiding detection by the defender. The attacker decides when to start an  intrusion and may stop the intrusion at any moment. During the intrusion, the attacker follows a pre-defined strategy. When deciding the time to start or stop an intrusion, the attacker considers both the gain of compromising additional servers and the risk of detection. The optimal strategy for the attacker is to compromise as many servers as possible without being detected.

The defender continuously monitors the infrastructure through accessing and analyzing \idps alerts and other statistics. It can take a fixed number of defensive actions, each of which has a cost and a chance of stopping an ongoing attack. An example of a defensive action is to drop network traffic that triggers \idps alerts of a certain priority. The defender takes actions in a pre-determined order, starting with the action that has the lowest cost. The final action blocks all external access to the gateway, which disrupts any intrusion as well as the services to the clients.

When deciding the time for taking a defensive action, the defender balances two objectives: (\textit{i}) maintain services to its clients; and (\textit{ii}) stop a possible intrusion at the lowest cost. The optimal strategy for the defender is to monitor the infrastructure and maintain services until the moment when the attacker enters through the gateway, at which time the attack must be stopped at minimal cost through defensive actions. The challenge for the defender is to identify this precise moment.

\section{Formalizing The Intrusion Response Use Case}\label{sec:formal_model_2}
We formulate the above intrusion response use case as a partially observed stochastic game. The attacker wins the game when it can intrude on the infrastructure and hide its actions from the defender. Similarly, the defender wins the game when it manages to stop an intrusion. It is a zero-sum game, which means that the gain of one player equals the loss of the other player.

The attacker and the defender have different observability in the game. The defender observes alerts from an Intrusion Detection and Prevention System (\idps) but has no certainty about the presence of an attacker or the state of a possible intrusion. The attacker, on the other hand, is assumed to have complete observability. It has access to all the information that the defender has access to, as well as the defender's past actions. This means that the defender has to find strategies that are effective against an opponent that has more knowledge than itself.

The reward function of the game encodes the defender's objective. An optimal defender strategy \textit{maximizes} the reward when facing an attacker with an optimal strategy, i.e. a worst-case attacker. Similarly, an optimal attacker strategy \textit{minimizes} the reward when facing a worst-case defender. Such a pair of optimal strategies is known as a Nash equilibrium in game theory \cite{nash51}.

We model the game as a finite, zero-sum Partially Observed Stochastic Game (\posg) with one-sided partial observability:
\begin{align}
\Gamma = \langle \mathcal{N}, \mathcal{S}, (\mathcal{A}_i)_{i \in \mathcal{N}}, \mathcal{T}, \mathcal{R}, \gamma,  \rho_1, T, \mathcal{O}, \mathcal{Z} \rangle \label{eq:game_def}
\end{align}
It is a discrete-time game that starts at time $t=1$ and ends at time $t=T$. In the following, we describe the components of the game, its evolution, and the objectives of the players.

\textbf{Players $\mathcal{N}$.} The game has two players: player $\mathrm{D}$ is the defender and player $\mathrm{A}$ is the attacker. Hence, $\mathcal{N}=\{\mathrm{D},\mathrm{A}\}$.

\textbf{Time horizon $T$.}
The time horizon $T$ is a random variable that depends on both players' strategies and takes values in the set $T \in \{2,3,\hdots,\infty\}$.

\textbf{State space $\mathcal{S}$.} The game has three states: $s_t=0$ if no intrusion occurs, $s_t=1$ if an intrusion is ongoing, and $s_T=\emptyset$ if the game has ended. Hence, $\mathcal{S}=\{0,1,\emptyset\}$. The state at time $t$, $s_t$, is a realization of the random variable $S_t$. The initial state is $s_1=0$. Hence the initial state distribution $\rho_1\colon \mathcal{S} \rightarrow [0,1]$ is the degenerate distribution $\rho_1(0)=1$.

\textbf{Action spaces $\mathcal{A}_i$.} Each player $i\in \mathcal{N}$ can invoke two actions: ``stop'' ($\mathfrak{S}$) and ``continue'' ($\mathfrak{C}$). The action spaces are thus $\mathcal{A}_{\mathrm{D}}=\mathcal{A}_{\mathrm{A}}=\{\mathfrak{S},\mathfrak{C}\}$. Executing action $\mathfrak{S}$ triggers a change in the game while action $\mathfrak{C}$ is a passive action. In the following, we encode $\mathfrak{S}$ with $1$ and $\mathfrak{C}$ with $0$.

The attacker can invoke the stop action twice: the first time to start the intrusion and the second time to terminate it.

The defender can invoke the stop action $L \geq 1$ times. A stop action is a defensive action against a possible intrusion. The number of stop actions remaining to the defender is known to both players and is denoted by $l \in \{1,\hdots,L\}$.

At each time-step, the attacker and the defender simultaneously choose an action $\mathbf{a}_t = (a^{(\mathrm{D})}_t, a^{(\mathrm{A})}_t)$, where $a^{(i)}_t \in \mathcal{A}_i$ is a realization of the random variable $A^{(i)}_t$ and $\mathbf{a}_t$ is a realization of the random vector $\mathbf{A}_t$.

\textbf{Observation space $\mathcal{O}$.} The attacker has complete observability and knows the game state, the defender's actions, and the defender's observations. In contrast, the defender has a limited set of observations $o_t \in \mathcal{O}$, where $\mathcal{O}$ is a discrete set. (In our use case, $o_t$ relates to the weighted sum of \idps alerts triggered during time-step $t$. We focus on the \idps alert metric as it provides more information than other possible metrics for detecting intrusions, see Appendix \ref{appendix:infrastructure_metrics} for details.)

Both players have perfect recall, meaning they remember their respective play history. The history of the defender at time-step $t$ is the vector $\mathbf{h}^{(\mathrm{D})}_t=\allowbreak(\rho_1, a^{(\mathrm{D})}_{1}\allowbreak, o_{1}\allowbreak, \hdots,\allowbreak a^{(\mathrm{D})}_{t-1},\allowbreak o_{t})$ and the history of the attacker is the vector $\mathbf{h}^{(\mathrm{A})}_t=\allowbreak(\rho_1,\allowbreak a^{(1)}_{1},\allowbreak a^{(\mathrm{A})}_{1},\allowbreak o_{1},\allowbreak s_1,\allowbreak \hdots,\allowbreak a^{(\mathrm{D})}_{t-1},\allowbreak a^{(\mathrm{A})}_{t-1},\allowbreak o_{t},\allowbreak s_{t})$, which are realizations of the random vectors $\mathbf{H}^{(\mathrm{D})}_t$ and $\mathbf{H}^{(\mathrm{A})}_t$, respectively.

\textbf{Belief space $\mathcal{B}$.} Based on its history $\mathbf{h}^{(\mathrm{D})}_t$, the defender forms a belief about the game state $s_t$, which is expressed in the \textit{belief state} $b_t(s_t)=\mathbb{P}[S_t=s_t\mid \mathbf{H}^{(\mathrm{D})}_t]$. Since $s_{t} \in \{0,1\}$ and $b_{t}(0) = 1-b_t(1)$ for $t<T$, we can model $\mathcal{B} = [0,1] \subset \mathbb{R}$, where $b_t(1) \in \mathcal{B}$ is a realization of the random variable $B_t$.

\textbf{Transition probabilities $\mathcal{T}$.} At each time-step $t$, a state transition from $s_t$ to $s_{t+1}$ occurs with probability $\mathcal{T}\big(s_{t+1}, s_t,\allowbreak (a^{(\mathrm{D})}_t, \allowbreak a^{(\mathrm{A})}_t)\big) = \allowbreak\mathbb{P}\big[S_{t+1}=s_{t+1}\mid \allowbreak S_t=s_t, \allowbreak \mathbf{A}_t=\allowbreak(a^{(\mathrm{D})}_t,\allowbreak a^{(\mathrm{A})}_t)\allowbreak\big]$:
\begin{align}
&\mathcal{T}_{l>1}\big(0,0,(\mathfrak{S},\mathfrak{C})\big)=\mathcal{T}\big(0,0,(\mathfrak{C},\mathfrak{C})\big)=1\label{eq:tp_1}\\
&\mathcal{T}_{l>1}\big(1,1,(\cdot,\mathfrak{C})\big)=\mathcal{T}\big(1,1,(\mathfrak{C},\mathfrak{C})\big) = 1-\phi_{l}  \label{eq:tp_2}\\
&\mathcal{T}_{l>1}\big(1,0,(\cdot,\mathfrak{S})\big)=\mathcal{T}\big(1,0,(\mathfrak{C},\mathfrak{S})\big) = 1\label{eq:tp_3}\\
&\mathcal{T}_{l>1}\big(\emptyset,1,(\cdot,\mathfrak{C})\big)=\mathcal{T}\big(\emptyset,1,(\mathfrak{C},\mathfrak{C})\big)=\phi_{l}\label{eq:tp_4}\\
&\mathcal{T}_{l=1}\big(\emptyset,\cdot,(\mathfrak{S},\cdot)\big)=\mathcal{T}(\emptyset,\emptyset,\cdot)=\mathcal{T}(\emptyset,1,(\cdot, \mathfrak{S}))=1\label{eq:tp_7}
\end{align}
where $\mathcal{T}_{l>1}$ and $\mathcal{T}_{l=1}$ refer to the transition probabilities when $l>1$ and $l=1$, respectively ($\mathcal{T}$ denotes the transition probabilities for any value of $l$). All other state transitions have probability $0$.

(\ref{eq:tp_1})--(\ref{eq:tp_2}) define the probabilities of the recurrent state transitions $0\rightarrow 0$ and $1\rightarrow 1$. The game stays in state $0$ with probability $1$ if the attacker selects action $\mathfrak{C}$ and $l_t-a_t^{(\mathrm{D})}>0$. (Note that $l_{t+1}=l_t-a_t^{(\mathrm{D})}$.) Similarly, the game stays in state $1$ with probability $1-\phi_{l}$ if the attacker chooses action $\mathfrak{C}$ and $l_t-a_t^{(\mathrm{D})}>0$. Here $\phi_{l}$ denotes the probability that the intrusion is stopped, which is a parameter of the use case. The intrusion can be stopped at any time-step as a consequence of previous stop actions by the defender. We assume that $\phi_{l}$ increases with each stop action that the defender takes.

(\ref{eq:tp_3}) captures the transition $0 \rightarrow 1$, which occurs when the attacker chooses action $\mathfrak{S}$ and $l_t-a_t^{(\mathrm{D})}>0$. (\ref{eq:tp_4})--(\ref{eq:tp_7}) define the probabilities of the transitions to the terminal state $\emptyset$. The terminal state is reached in three cases: (\textit{i}) when $l_t=1$ and the defender takes the final stop action $\mathfrak{S}$ (i.e. when $l_t-a^{(\mathrm{D})}_t=0$); (\textit{ii}) when the intrusion is stopped by the defender with probability $\phi_{l}$; and (\textit{iii}) when $s_t=1$ and the attacker terminates the intrusion ($a^{(\mathrm{A})}_t=1$).

The evolution of the game can be described with the state transition diagram in Fig. \ref{fig:state_transitions}. The figure captures a game \textit{episode}, which starts at $t=1$ and ends at $t=T$.

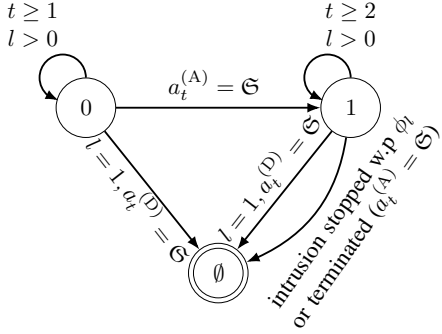
\begin{figure}
  \centering
  \scalebox{1.1}{
    \begin{tikzpicture}[fill=white, >=stealth,
    node distance=3cm,
    database/.style={
      cylinder,
      cylinder uses custom fill,
      shape border rotate=90,
      aspect=0.25,
      draw}]

\node[scale=0.8] (kth_cr) at (0,2.15)
{
  \begin{tikzpicture}

\node[scale=1] (level1) at (-1.7,-5.6)
{
  \begin{tikzpicture}
\node[draw,circle, minimum width=15mm, scale=0.6](s0) at (0,0) {};
\node[draw,circle, minimum width=15mm, scale=0.6](s1) at (4,0) {};
\node[draw,circle, minimum width=15mm, scale=0.6](s2) at (2,-2.5) {};
\node[draw,circle, minimum width=15mm, scale=0.5](s4) at (2,-2.5) {};

\node[inner sep=0pt,align=center, scale=1] (time) at (0.07,0)
{
$0$
};

\node[inner sep=0pt,align=center, scale=1] (time) at (4.07,0)
{
$1$
};

\node[inner sep=0pt,align=center, scale=1] (time) at (2.07,-2.5)
{
$\emptyset$
};

\node[inner sep=0pt,align=center, scale=1] (time) at (-0.8,1.25)
{
$t\geq 1$\\
$l > 0$
};

\node[inner sep=0pt,align=center, scale=1] (time) at (4,1.25)
{
$t\geq 2$\\
$l > 0$
};

\node[inner sep=0pt,align=center, scale=1] (time) at (2,0.35)
{
$a^{(\mathrm{A})}_t=\mathfrak{S}$
};

\node[inner sep=0pt,align=center, scale=1,rotate=-52] (time) at (0.84,-1.45)
{
  $l=1, a_t^{(\mathrm{D})}=\mathfrak{S}$
};

\node[inner sep=0pt,align=center, scale=1,rotate=52] (time) at (2.8,-1)
{
  $l=1, a_t^{(\mathrm{D})}=\mathfrak{S}$
};

\node[inner sep=0pt,align=center, scale=1,rotate=55] (time) at (4.05,-1.68)
{
  intrusion stopped w.p $\phi_l$\\
  or terminated ($a^{(\mathrm{A})}_t=\mathfrak{S}$)
};


\draw[thick,-{Latex[length=2mm]}] (s0) to (s1);
\draw[thick,-{Latex[length=2mm]}] (s1) to (s2);
\draw[thick,-{Latex[length=2mm]}] (s0) to (s2);

\draw[thick,-{Latex[length=2mm]}, bend left=30] (s1) to (s2);

\draw[thick,-{Latex[length=2mm]}] (s0.90) arc (0:260:3.5mm);

\draw[thick,-{Latex[length=2mm]}] (s1.90) arc (0:260:3.5mm);

    \end{tikzpicture}
  };
    \end{tikzpicture}
  };

\end{tikzpicture}
    }
    \caption{State transition diagram of a game episode: each disk represents a state; an arrow represents a state transition; a label indicates the conditions for the state transition (w.p means ``with probability''); a game episode starts in state $s_1=0$ with $l=L$ and ends in state $s_T=\emptyset$.}
    \label{fig:state_transitions}
  \end{figure}

\textbf{Reward function $\mathcal{R}$.} At time-step $t$, the defender receives the reward $r_t = \mathcal{R}(s_t,(a^{(\mathrm{D})}_t, a^{(\mathrm{A})}_t))$ and the attacker receives the reward $-r_t$. Here $r_t \in \mathbb{R}$ is a realization of the random variable $R_t$. The reward function $\mathcal{R}$ is parameterized by the defender's reward for stopping an intrusion ($\mathrm{R}_{\mathrm{st}} \in \mathbb{R}_{>0}$), the defender's cost of taking a defensive action ($\mathrm{R}_{\mathrm{cost}} \in \mathbb{R}_{<0}$), and the defender's cost while an intrusion occurs ($\mathrm{R}_{\mathrm{int}} \in \mathbb{R}_{<0}$):
\begin{align}
&\mathcal{R}(\emptyset, \cdot)=0\label{eq:reward_0}\\
&\mathcal{R}\big(1, (\cdot,\mathfrak{S})\big) = 0 \label{eq:reward_1}\\
&\mathcal{R}\big(0, (\mathfrak{C},\cdot)\big)=0 \label{eq:reward_2}\\
&\mathcal{R}\big(0, (\mathfrak{S},\cdot)\big)=\frac{\mathrm{R}_{\mathrm{cost}}}{l_t}    \label{eq:reward_3}\\
&\mathcal{R}\big(1, (\mathfrak{S},\mathfrak{C})\big)=\frac{\mathrm{R}_{\mathrm{st}}}{l_t}  \label{eq:reward_4}\\
&\mathcal{R}\big(1, (\mathfrak{C},\mathfrak{C})\big)=\mathrm{R}_{\mathrm{int}}  \label{eq:reward_5}
\end{align}
(\ref{eq:reward_0})--(\ref{eq:reward_1}) state that the reward is zero in the terminal state and when the attacker terminates an intrusion. (\ref{eq:reward_2}) states that the defender incurs no cost when no attack occurs and it does not take a defensive action. (\ref{eq:reward_3}) indicates that the defender incurs a cost when taking a defensive action if no intrusion is ongoing. (\ref{eq:reward_4}) states that the defender receives a reward when taking a stop action while an intrusion occurs. Lastly, (\ref{eq:reward_5}) indicates that the defender incurs a cost for each time-step during which an intrusion occurs.

\textbf{Observation function $\mathcal{Z}$.} At time-step $t$, $o_t \in \mathcal{O}$ is drawn from a random variable $O$ whose distribution $f_{O}$ depends on the current state $s_{t}$. We define $\mathcal{Z}(o_{t},s_t,(a^{(\mathrm{D})}_{t-1}, a^{(\mathrm{A})}_{t-1}))=\allowbreak \mathbb{P}[O=o_t\mid \allowbreak S_t=s_t, \allowbreak \mathbf{A}_t=(a^{(\mathrm{D})}_{t-1},\allowbreak a^{(\mathrm{A})}_{t-1})\allowbreak ]$ as follows:
\begin{align}
&\mathcal{Z}\big(o_t,0,\cdot \big) = f_{O}(o_t \mid 0) \label{eq:obs_1}\\
&\mathcal{Z}\big(o_t,1,\cdot \big) = f_{O}(o_t \mid 1)\label{eq:obs_2}
\end{align}\normalsize

\textbf{Player strategies $\pi_{i}$.} A defender strategy is a function $\pi_{\mathrm{D}} \in \Pi_{\mathrm{D}} = \{1,\hdots, L\} \times \mathcal{B} \rightarrow \Delta(\mathcal{A}_{\mathrm{D}})$, where $\Delta(\mathcal{A}_{\mathrm{D}})$ denotes the set of probability distributions over $\mathcal{A}_{\mathrm{D}}$. Similarly, an attacker strategy is a function $\pi_{\mathrm{A}} \in \Pi_{\mathrm{A}} = \{1,\hdots, L\} \times \mathcal{B} \times \mathcal{S} \rightarrow \Delta(\mathcal{A}_{\mathrm{A}})$. The strategies for both players are dependent on $l$ but independent of $t$ (i.e. strategies are stationary). If $\pi_{i}$ always maps on to an action with probability $1$, it is called \textit{pure}, otherwise it is called \textit{mixed}. In other words, a pure strategy is deterministic and a mixed strategy is stochastic.

\textbf{Belief update.} At time-step $t>1$, the defender updates the belief state $b_{t-1}$ using the equation
\begin{align}
&b_{t}(s_{t}) = C \sum_{s_{t-1} \in \mathcal{S}}\sum_{a^{(\mathrm{A})}_{t-1} \in \mathcal{A}_{\mathrm{A}}}\Big(b_{t-1}(s_{t-1})\pi_{\mathrm{A}}(a^{(\mathrm{A})}_{t-1}\mid s_{t-1},b_{t-1})\cdot \nonumber\\
  &\mathcal{Z}(o_{t},s_{t},(a^{(\mathrm{D})}_{t-1},a^{(\mathrm{A})}_{t-1})) \mathcal{T}\big(s_{t},s_{t-1},(a^{(\mathrm{D})}_{t-1},a^{(\mathrm{A})}_{t-1})\big)\Big)\label{eq:belief_upd}
\end{align}
where $C=1/\mathbb{P}[o_{t}\mid a^{(\mathrm{D})}_{t-1},\pi_{\mathrm{A}}, b_{t-1}]$ is a normalizing factor to ensure that the sum over $b_{t}(s_t)$ for all $s_t$ equals $1$. The initial belief is $b_1(0)=1$.

\textbf{Objective functions $J_i$.} The goal of the defender is to \textit{maximize} the expected discounted cumulative reward over the time horizon $T$. Similarly, the goal of the attacker is to \textit{minimize} the same quantity. Therefore, the objective functions $J_{\mathrm{D}}$ and $J_{\mathrm{A}}$ are
\begin{align}
J_{\mathrm{D}}(\pi_{\mathrm{D}}, \pi_{\mathrm{A}}) &= \mathbb{E}_{(\pi_{\mathrm{D}}, \pi_{\mathrm{A}})}\left[\sum_{t=1}^{T}\gamma^{t-1}\mathcal{R}(S_t, \mathbf{A}_t)\right] \label{eq:objective_1}\\
J_{\mathrm{A}}(\pi_{\mathrm{D}}, \pi_{\mathrm{A}}) &= -J_{\mathrm{D}}(\pi_{\mathrm{D}}, \pi_{\mathrm{A}}) \label{eq:objective_2}
\end{align}
where $\gamma \in [0,1)$ is the discount factor and $\mathbb{E}_{(\pi_{\mathrm{D}}, \pi_{\mathrm{A}})}$ denotes the expectation of the random variables $(S_t,O_t,\mathbf{A}_t)_{t \in \{1,\hdots,T\}}$ under strategy profile $(\pi_{\mathrm{D}}, \pi_{\mathrm{A}})$.

\textbf{Best response strategies $\tilde{\pi}_{i}$.} A defender strategy $\tilde{\pi}_{\mathrm{D}} \in \Pi_{\mathrm{D}}$ is called a \textit{best response} against $\pi_{\mathrm{A}}\in \Pi_{\mathrm{A}}$ if it \textit{maximizes} $J_{\mathrm{D}}$ (\ref{eq:objective_1}). Similarly, an attacker strategy $\tilde{\pi}_{\mathrm{A}}$ is called a best response against $\pi_{\mathrm{D}}$ if it \textit{minimizes} $J_{\mathrm{D}}$ (\ref{eq:objective_2}). Hence, the best response correspondences $\mathscr{B}_{\mathrm{D}}$ and $\mathscr{B}_{\mathrm{A}}$ are obtained as follows:
\begin{align}
\mathscr{B}_{\mathrm{D}}(\pi_{\mathrm{A}}) &= \argmax_{\pi_{\mathrm{D}} \in \Pi_{\mathrm{D}}}J_{\mathrm{D}}(\pi_{\mathrm{D}}, \pi_{\mathrm{A}})\label{eq:br_defender}\\
\mathscr{B}_{\mathrm{A}}(\pi_{\mathrm{D}}) &= \argmin_{\pi_{\mathrm{A}} \in \Pi_{\mathrm{A}}}J_{\mathrm{D}}(\pi_{\mathrm{D}}, \pi_{\mathrm{A}})\label{eq:br_attacker}
\end{align}

\textbf{Optimal strategies $\pi^{*}_{i}$.} An optimal defender strategy $\pi_{\mathrm{D}}^{*}$ is a best response strategy against any attacker strategy that \textit{minimizes} $J_{\mathrm{D}}$. Similarly, an optimal attacker strategy $\pi_{\mathrm{A}}^{*}$ is a best response against any defender strategy that \textit{maximizes} $J_{\mathrm{D}}$. Hence, when both players follow optimal strategies, they play best response strategies against each other:
\begin{align}
(\pi_{\mathrm{D}}^{*}, \pi_{\mathrm{A}}^{*}) \in \mathscr{B}_{\mathrm{D}}(\pi_{\mathrm{A}}^{*}) \times \mathscr{B}_{\mathrm{A}}(\pi_{\mathrm{D}}^{*})\label{eq:minmax_objective}
\end{align}
Since no player has an incentive to change its strategy, $(\pi_{\mathrm{D}}^{*},\pi_{\mathrm{A}}^{*})$ is a Nash equilibrium \cite{nash51}.

\begin{table}
  \centering
  \begin{tabular}{ll} \toprule
    {\textit{Notation(s)}} & {\textit{Description}} \\ \midrule
    $\Gamma$ & The intrusion response \posg (\ref{eq:game_def})\\
    $\mathrm{D}, \mathrm{A}$ & The defender player and the attacker player\\
    $t, T, \gamma$ & Time-step, time horizon, and discount factor\\
    $l_t$ & Defender stops remaining at time-step $t$\\
    $L$ & Maximum number of defender stops\\
    $\pi_{\mathrm{D}},\pi_{\mathrm{A}}$ & Defender and attacker strategies\\
    $\tilde{\pi}_{\mathrm{D}},\tilde{\pi}_{\mathrm{A}}$ & Best response strategies\\
    $\pi^{*}_{\mathrm{D}},\pi^{*}_{\mathrm{A}}$ & Optimal strategies\\
    $\mathcal{N},\mathcal{S},\mathcal{O}$ & Sets of players, states, and observations\\
    $\mathcal{A}_{\mathrm{D}},\mathcal{A}_{\mathrm{A}}$ & Sets of defender and attacker actions\\
    $\mathcal{T}, \mathcal{R}, \mathcal{Z}$ & Transition, reward and observation functions\\
    $s_t, o_t, \mathbf{a}_t=(a_t^{\mathrm{D}}, a_t^{\mathrm{A}})$ & State, observation, and actions at time-step $t$\\
    $b_t(1) \in \mathcal{B}, r_t$ & Defender belief and reward at time-step $t$\\
    $\mathfrak{S}, \mathfrak{C}$ & Stop and continue actions\\
    $\tau_{i,k}$ & $k$th stopping time of player $i$\\
    $\mathscr{B}_{\mathrm{D}}, \mathscr{B}_{\mathrm{A}}$ & Best response correspondences (\ref{eq:br_defender})--(\ref{eq:br_attacker})\\
    $J_{\mathrm{D}}, J_{\mathrm{A}}$ & Defender and attacker objectives (\ref{eq:objective_1})--(\ref{eq:objective_2})\\
    $\mathcal{M}^P, \mathcal{M}$ & Best response \pomdp and \mdp for $\mathrm{D}$ and $\mathrm{A}$\\
    $\mathscr{S}^{i}, \mathscr{C}^{i}$ & Stopping and continuation sets of player $i$\\
    $S_t,\mathbf{A}_{t},O_t$ & Random variables with realizations $s_t,\mathbf{a}_t,o_t$\\
    $R_t, B_t$ & Random variables with realizations $r_t,b_t$\\
    $V^{*}_{l,\pi_{\mathrm{A}}},V^{*}_{l,\pi_{\mathrm{D}}}$ & Value functions of $\mathcal{M}^P$ and $\mathcal{M}$\\
    $f_{O\mid s}$ & Observation distribution (\ref{eq:obs_1})--(\ref{eq:obs_2})\\
    \bottomrule\\
  \end{tabular}
  \caption{Notations for our mathematical model.}\label{tab:notation}
\end{table}
\section{Game-Theoretic Analysis and Our Algorithm for Finding Near-Optimal Defender Strategies}\label{sec:game_analysis}
Finding optimal strategies that satisfy (\ref{eq:minmax_objective}) is equivalent to finding a Nash equilibrium for the \posg $\Gamma$ (\ref{eq:game_def}). We know from game theory that $\Gamma$ has at least one mixed Nash equilibrium \cite{vonNeumann_1928:TGG,nash51, Shapley1095,posg_equilibria_existence_finite_horizon}. (A Nash equilibrium is called mixed if one or more players follow mixed strategies.) In this section, we first analyze the structure of Nash equilibria in $\Gamma$ using optimal stopping theory and then we describe an efficient reinforcement learning algorithm for approximating these equilibria.

\subsection{Analyzing Best Responses using Optimal Stopping Theory}
The equilibria in $\Gamma$ can be obtained by finding the pairs of strategies that are best responses against each other (\ref{eq:minmax_objective}). A best response for the defender is obtained by solving a \pomdp $\mathcal{M}^{P}$, and a best response for the attacker is obtained by solving an \mdp $\mathcal{M}$. The corresponding Bellman equations are \cite{bellman1957markovian}:
\begin{align}
&V_{l,\pi_{\mathrm{A}}}^{*}(b_t) = \max_{a^{(\mathrm{D})}_t\in \mathcal{A}_{\mathrm{D}}} \mathop{\mathbb{E}}_{\pi_{\mathrm{A}},b_t,a^{(\mathrm{D})}_t}\Big[R_{t+1} + \gamma V_{l-a^{(\mathrm{D})}_t,\pi_{\mathrm{A}}}^{*}(B_{t+1})\Big]\label{eq:bellman_eq_41}\\
&V_{l,\pi_{\mathrm{D}}}^{*}((b_t,s_t)) =\nonumber\\
&\min_{a^{(\mathrm{A})}_t\in \mathcal{A}_{\mathrm{A}}} \mathop{\mathbb{E}}_{\pi_{\mathrm{D}},a^{(\mathrm{A})}_t}\Big[R_{t+1} + \gamma V_{l-A^{(\mathrm{D})}_t,\pi_{\mathrm{D}}}^{*}((B_{t+1}, S_{t+1}))\Big]\label{eq:bellman_eq_43}
\end{align}
where $V_{l,\pi_{\mathrm{A}}}^{*}$ is the value function in the \pomdp $\mathcal{M}^{P}$ given that the attacker follows strategy $\pi_{\mathrm{A}}$ and the defender has $l$ stops remaining, and $V^{*}_{l,\pi_{\mathrm{D}}}$ is the value function in the \mdp $\mathcal{M}$ given that the defender follows strategy $\pi_{\mathrm{D}}$ and has $l$ stops remaining.

Since the game is zero-sum, stationary, and $\gamma < 1$, it follows from the Minimax theorem in game theory that there exists a value function:
\begin{align}
  &V_l^{*}(b_t) = \max_{\pi_{\mathrm{D}} \in \Delta(\mathcal{A}_{\mathrm{D}})}\min_{\pi_{\mathrm{A}} \in \Delta(\mathcal{A}_{\mathrm{A}})} \mathop{\mathbb{E}}_{\pi_{\mathrm{D}},\pi_{\mathrm{A}},b_t}\Big[R_{t+1} +  \label{eq:bellman_posg_1}\\
 &\quad\quad\quad\quad\quad\quad\quad\quad\quad\quad\quad\quad\quad\quad\quad\quad \gamma V_{l-a^{(\mathrm{D})}_t}^{*}(B_{t+1})\Big] \nonumber
\end{align}
and that $V_l^{*}(b) = V_{l,\pi^{*}_{\mathrm{A}}}^{*}(b) = V_{l,\pi^{*}_{\mathrm{D}}}^{*}(b,s)$ \cite{vonNeumann_1928:TGG}\cite[Thm. 2.3]{horak_thesis}. Further, from Markov decision theory we know that for any strategy pair ($\pi_{\mathrm{D}}, \pi_{\mathrm{A}}$), a corresponding pair of \textit{pure} best response strategies $(\tilde{\pi}_{\mathrm{D}}, \tilde{\pi}_{\mathrm{A}}) \in \mathscr{B}_{\mathrm{D}}(\pi_{\mathrm{A}}) \times \mathscr{B}_{\mathrm{A}}(\pi_{\mathrm{D}})$ exists \cite[Thm. 6.2.7]{puterman}\cite[Thm. 7.6.1-7.6.2]{krishnamurthy_2016}.
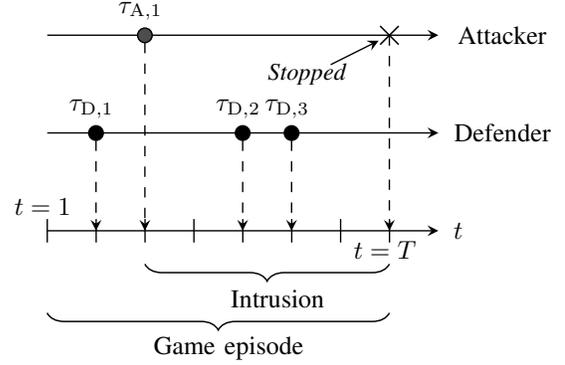
\begin{figure}
  \centering
  \scalebox{1.3}{
    \begin{tikzpicture}[fill=white, >=stealth,
    node distance=3cm,
    database/.style={
      cylinder,
      cylinder uses custom fill,
      shape border rotate=90,
      aspect=0.25,
      draw}]

    \tikzset{
node distance = 9em and 4em,
sloped,
   box/.style = {%
    shape=rectangle,
    rounded corners,
    draw=blue!40,
    fill=blue!15,
    align=center,
    font=\fontsize{12}{12}\selectfont},
 arrow/.style = {%
    line width=0.1mm,
    -{Triangle[length=5mm,width=2mm]},
    shorten >=1mm, shorten <=1mm,
    font=\fontsize{8}{8}\selectfont},
}

\node[scale=1] (system) at (0,-3)
{
\begin{tikzpicture}
\draw[->, color=black] (0, 0) to (4, 0);
\draw[-, color=black] (0, -0.12) to (0, 0.12);
\draw[-, color=black] (0.5, -0.12) to (0.5, 0.12);
\draw[-, color=black] (1, -0.12) to (1, 0.12);
\draw[-, color=black] (1.5, -0.12) to (1.5, 0.12);
\draw[-, color=black] (2, -0.12) to (2, 0.12);
\draw[-, color=black] (2.5, -0.12) to (2.5, 0.12);
\draw[-, color=black] (3, -0.12) to (3, 0.12);
\draw[-, color=black] (3.5, -0.12) to (3.5, 0.12);
\end{tikzpicture}
};

\node[scale=1] (system) at (0,-1)
{
\begin{tikzpicture}
\draw[->, color=black] (0, 0) to (4, 0);
\end{tikzpicture}
};

\node[scale=1] (system) at (0,-2)
{
\begin{tikzpicture}
\draw[->, color=black] (0, 0) to (4, 0);

\end{tikzpicture}
};

\node[inner sep=0pt,align=center, scale=0.75] (time) at (2.7,-1)
{
Attacker
};

\node[inner sep=0pt,align=center, scale=0.75] (time) at (2.7,-2)
{
Defender
};

\node[inner sep=0pt,align=center, scale=0.75] (time) at (-2,-2.75)
{
$t=1$
};

\node[inner sep=0pt,align=center, scale=0.75] (time) at (1.5,-3.2)
{
$t=T$
};

\node[inner sep=0pt,align=center, scale=0.75] (time) at (-1.5,-1.75)
{
$\tau_{\mathrm{D},1}$
};

\node[inner sep=0pt,align=center, scale=0.75] (time) at (0.0,-1.75)
{
$\tau_{\mathrm{D},2}$
};
\node[inner sep=0pt,align=center, scale=0.75] (time) at (0.5,-1.75)
{
$\tau_{\mathrm{D},3}$
};

\node[inner sep=0pt,align=center, scale=0.75] (time) at (-1,-0.75)
{
$\tau_{\mathrm{A},1}$
};


\node[inner sep=0pt,align=center, scale=0.75] (time) at (2.25,-3)
{
$t$
};

\node[draw,circle, minimum width=2mm, color=black, fill=black!70, scale=0.45](a2) at (-1,-1) {};
\draw[->, color=black, dashed] (a2) to (-1, -3);
\node[draw,circle, minimum width=2mm, color=black, scale=0.45, fill=black](d1) at (-1.5,-2) {};

\node[draw,circle, minimum width=2mm, color=black, scale=0.45, fill=black](d3) at (0.0,-2) {};
\node[draw,circle, minimum width=2mm, color=black, scale=0.45, fill=black](d4) at (0.5,-2) {};

\draw[->, color=black, dashed] (d1) to (-1.5, -3);
\draw[->, color=black, dashed] (d3) to (0.0, -3);
\draw[->, color=black, dashed] (d4) to (0.5, -3);

\draw[->, color=black, dashed] (1.5,-1) to (1.5, -3);

\node[inner sep=0pt,align=center, scale=1] (prev) at (1.5,-1)
{
\Crossk
};
\node[inner sep=0pt,align=center, scale=0.7] (time) at (0.7,-1.4)
{
  \textit{Stopped}
};



\draw[->, color=black] (0.9, -1.3) to (prev);

\node[inner sep=0pt,align=left, scale=0.75] (time) at (-0.1,-4.2)
{
Game episode
};

\node[inner sep=0pt,align=left, scale=0.75] (time) at (0.4,-3.7)
{
Intrusion
};

\draw [decorate,decoration={brace,amplitude=5pt,mirror,raise=4pt},yshift=0pt,rotate=180, line width=0.15mm]
(2,3.7) -- (-1.5,3.7) node [black,midway,xshift=0.1cm] {};

\draw [decorate,decoration={brace,amplitude=5pt,mirror,raise=4pt},yshift=0pt,rotate=180, line width=0.15mm]
(1,3.2) -- (-1.5,3.2) node [black,midway,xshift=0.1cm] {};

\end{tikzpicture}
  }
    \caption{Stopping times of the defender and the attacker in a game episode; the bottom horizontal axis represents time; the black circles on the middle axis and the upper axis represent time-steps of the defender's stop actions and the attacker's stop actions, respectively; $\tau_{i,j}$ denotes the $j$th stopping time of player $i$; the cross shows the time the intrusion is stopped; an intrusion starts when the attacker takes the first stop action (at time $\tau_{\mathrm{A},1}$); an episode ends either when the attacker is stopped (as a consequence of defender actions) or when the attacker terminates its intrusion.}
    \label{fig:stopping_times}
  \end{figure}

  We interpret the \pomdp $\mathcal{M}^{P}$ and the \mdp $\mathcal{M}$ that determine the best response strategies as \textit{optimal stopping} problems (see Fig. \ref{fig:stopping_times}) \cite{wald,stopping_book_1,chow1971great,hammar_stadler_tnsm}. Consequently, an optimal solution to $\mathcal{M}^{P}$ (or $\mathcal{M}$) is also an optimal solution to the corresponding stopping problem and vice versa.

The problem for the defender is to find a stopping strategy $\pi_{\mathrm{D}}^{*}(b_t) \rightarrow \{\mathfrak{S},\mathfrak{C}\}$ that maximizes $J_{\mathrm{D}}$ (\ref{eq:objective_1}) and prescribes the optimal stopping times $\tau^{*}_{\mathrm{D},1},\tau^{*}_{\mathrm{D},2},\hdots, \tau^{*}_{\mathrm{D},L}$. Similarly, the problem for the attacker is to find a stopping strategy $\pi_{\mathrm{A}}^{*}(s_t,b_t) \rightarrow \{\mathfrak{S},\mathfrak{C}\}$ that minimizes $J_{\mathrm{D}}$ (\ref{eq:objective_2}) and prescribes the optimal stopping times $\tau^{*}_{\mathrm{A},1}$ and $\tau^{*}_{\mathrm{A},2}$.

Given a pair of stopping strategies $(\pi_{\mathrm{D}},\pi_{\mathrm{A}})$ and their (pure) best responses $\tilde{\pi}_{\mathrm{D}} \in \mathscr{B}_{\mathrm{D}}(\pi_{\mathrm{A}})$ and $\tilde{\pi}_{\mathrm{A}} \in \mathscr{B}_{\mathrm{A}}(\pi_{\mathrm{D}})$, we define two subsets of $\mathcal{B}=[0,1]$: the \textit{stopping sets} and the \textit{continuation sets}.

The stopping sets $\mathscr{S}^{(\mathrm{D})}$ and $\mathscr{S}^{(\mathrm{A})}$ include the belief states where $\mathfrak{S}$ is a best response:
\begin{align}
\mathscr{S}^{(\mathrm{D})}_{l,\pi_{\mathrm{A}}} &= \left\{b(1) \mid b(1) \in [0,1], \tilde{\pi}_{\mathrm{D},l}\big(b(1)\big) = \mathfrak{S}\right\}\label{eq:stopping_set_def}\\
\mathscr{S}^{(\mathrm{A})}_{s,l,\pi_{\mathrm{D}}} &= \left\{b(1) \mid b(1) \in [0,1], \tilde{\pi}_{\mathrm{A},l}\big(s,b(1)\big) = \mathfrak{S}\right\}\label{eq:stopping_set_at}
\end{align}
Similarly, the continuation sets $\mathscr{C}^{(\mathrm{D})}$ and $\mathscr{C}^{(\mathrm{A})}$ contain the belief states where $\mathfrak{C}$ is a best response:
\begin{align}
\mathscr{C}^{(\mathrm{D})}_{l,\pi_{\mathrm{A}}} &= \left\{b(1) \mid b(1) \in [0,1], \tilde{\pi}_{\mathrm{D},l}\big(b(1)\big) = \mathfrak{C}\right\}\label{eq:continue_set_def}\\
\mathscr{C}^{(\mathrm{A})}_{s,l,\pi_{\mathrm{D}}} &= \left\{b(1) \mid b(1) \in [0,1], \tilde{\pi}_{\mathrm{A},l}\big(b(1), s\big) = \mathfrak{C}\right\}\label{eq:continue_set_at}
\end{align}
Based on \cite[Thm. 12.3.4]{krishnamurthy_2016} \cite[Prop. 4.5-4.8]{Nakai1985}, \cite[Thm. 1]{optimal_multiple_stopping_social_media_1}\cite[Thm. 1]{hammar_stadler_tnsm}, and \cite[Thm. 2.3]{horak_solving_one_sided_posgs}, we formulate Theorem \ref{thm:best_responses}, which contains an existence result for equilibria and a structural result for best response strategies of the game.
\begin{theorem}\label{thm:best_responses}
Given the \posg $\Gamma$ (\ref{eq:game_def}) with one-sided partial observability and $L \geq 1$, the following holds:
\begin{enumerate}[(A)]
\item $\Gamma$ has a mixed Nash equilibrium. If $s=0 \iff b(1)=0$, then it has a pure Nash equilibrium.
\item We assume that the probability mass function $f_{O \mid s}$ is totally positive of order 2 (i.e., \tpp \cite[Def. 10.2.1, pp. 223]{krishnamurthy_2016}). Given an attacker strategy $\pi_{\mathrm{A}} \in \Pi_{\mathrm{A}}$, then there exist values $\tilde{\alpha}_{1}$ $\geq$ $\tilde{\alpha}_{2}$ $\geq$ $\hdots$ $\geq$ $\tilde{\alpha}_L \in [0,1]$ and a best response strategy $\tilde{\pi}_{\mathrm{D}} \in \mathscr{B}_{\mathrm{D}}(\pi_{\mathrm{A}})$ for the defender that satisfies
\begin{align}
\tilde{\pi}_{\mathrm{D},l}(b(1)) = \mathfrak{S} \iff b(1) \geq \tilde{\alpha}_l \quad l\in \{1,\hdots,L\} \label{eq:prop_br_defender}
\end{align}
\item Given a defender strategy $\pi_{\mathrm{D}}\in \Pi_{\mathrm{D}}$ where $\pi_{\mathrm{D}}(\mathfrak{S} \mid b(1))$ is non-decreasing in $b(1)$ and $\pi_{\mathrm{D}}(\mathfrak{S} \mid 1)=1$, then there exist values $\tilde{\beta}_{0,1},$ $\tilde{\beta}_{1,1},$ $\hdots$, $\tilde{\beta}_{0,L}$, $\tilde{\beta}_{1,L} \in [0,1]$ and a best response strategy $\tilde{\pi}_{\mathrm{A}} \in \mathscr{B}_{\mathrm{A}}(\pi_{\mathrm{D}})$ for the attacker that satisfies
  \begin{align}
\tilde{\pi}_{\mathrm{A},l}(b(1), 0) = \mathfrak{C} &\iff \pi_{\mathrm{D},l}(\mathfrak{S} \mid b(1)) \geq \tilde{\beta}_{0,l} \label{eq:prop_br_attacker_1}\\
\tilde{\pi}_{\mathrm{A},l}(b(1), 1) = \mathfrak{S} &\iff \pi_{\mathrm{D},l}(\mathfrak{S} \mid b(1)) \geq \tilde{\beta}_{\mathrm{D},l} \label{eq:prop_br_attacker_2}
\end{align}
for $l \in \{1,\hdots, L\}$.
\end{enumerate}
\end{theorem}
\begin{proof}[Proof.]
See Appendix \ref{appendix:proofs}.
\end{proof}
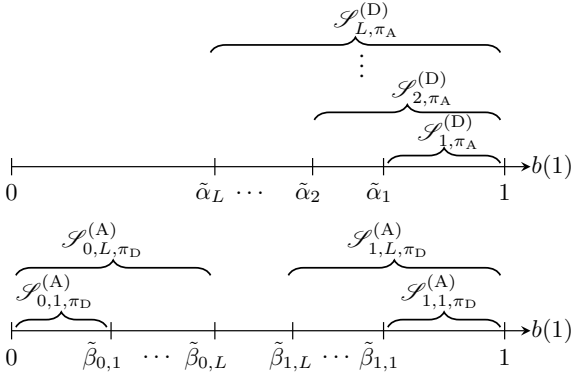
\begin{figure}
  \centering
  \scalebox{1.15}{
    \begin{tikzpicture}[fill=white, >=stealth,
    node distance=3cm,
    database/.style={
      cylinder,
      cylinder uses custom fill,
      shape border rotate=90,
      aspect=0.25,
      draw}]

    \tikzset{
node distance = 9em and 4em,
sloped,
   box/.style = {%
    shape=rectangle,
    rounded corners,
    draw=blue!40,
    fill=blue!15,
    align=center,
    font=\fontsize{12}{12}\selectfont},
 arrow/.style = {%
    line width=0.1mm,
    -{Triangle[length=5mm,width=2mm]},
    shorten >=1mm, shorten <=1mm,
    font=\fontsize{8}{8}\selectfont},
}

\node[scale=1] (system) at (0,0)
{
\begin{tikzpicture}
\draw[->, color=black] (0.0,0) to (6,0);

\node[inner sep=0pt,align=center, scale=0.8] (time) at (6.3,0)
{
  $b(1)$
};

\node[inner sep=0pt,align=center, scale=0.8] (time) at (0.05,-0.3)
{
$0$
};

\node[inner sep=0pt,align=center, scale=0.8] (time) at (5.75,-0.3)
{
$1$
};

\draw[-, color=black] (5.7,0.1) to (5.7,-0.1);

\draw[-, color=black] (0,0.1) to (0,-0.1);

\draw [decorate,decoration={brace,amplitude=5pt,mirror,raise=4pt},yshift=0pt,rotate=180, line width=0.20mm]
(-5.65,0.1) -- (-4.35,0.1) node [black,midway,xshift=0.1cm] {};

\node[inner sep=0pt,align=center, scale=0.8] (time) at (5.1,0.4)
{
$\mathscr{S}^{(\mathrm{D})}_{1,\pi_{\mathrm{A}}}$
};

\draw [decorate,decoration={brace,amplitude=5pt,mirror,raise=4pt},yshift=0pt,rotate=180, line width=0.20mm]
(-5.65,-0.4) -- (-3.5,-0.4) node [black,midway,xshift=0.1cm] {};

\node[inner sep=0pt,align=center, scale=0.8] (time) at (4.8,0.9)
{
$\mathscr{S}^{(\mathrm{D})}_{2,\pi_{\mathrm{A}}}$
};

\node[inner sep=0pt,align=center, scale=0.8] (time) at (4.1,1.25)
{
$\vdots$
};

\draw [decorate,decoration={brace,amplitude=5pt,mirror,raise=4pt},yshift=0pt,rotate=180, line width=0.20mm]
(-5.65,-1.18) -- (-2.3,-1.18) node [black,midway,xshift=0.1cm] {};

\node[inner sep=0pt,align=center, scale=0.8] (time) at (4.15,1.7)
{
$\mathscr{S}^{(\mathrm{D})}_{L,\pi_{\mathrm{A}}}$
};

\draw[-, color=black] (4.3,0.1) to (4.3,-0.1);
\draw[-, color=black] (3.48,0.1) to (3.48,-0.1);

\draw[-, color=black] (2.35,0.1) to (2.35,-0.1);

\node[inner sep=0pt,align=center, scale=0.8] (time) at (4.3,-0.3)
{
$\tilde{\alpha}_{1}$
};
\node[inner sep=0pt,align=center, scale=0.8] (time) at (3.48,-0.3)
{
$\tilde{\alpha}_{2}$
};

\node[inner sep=0pt,align=center, scale=0.8] (time) at (2.35,-0.3)
{
$\tilde{\alpha}_{L}$
};

\node[inner sep=0pt,align=center, scale=0.8] (time) at (2.8,-0.3)
{
$\hdots$
};
\end{tikzpicture}
};

\node[scale=1] (system) at (0,-2.25)
{
\begin{tikzpicture}
\draw[->, color=black] (0.0,0) to (6,0);

\node[inner sep=0pt,align=center, scale=0.8] (time) at (6.3,0)
{
  $b(1)$
};

\node[inner sep=0pt,align=center, scale=0.8] (time) at (0.05,-0.3)
{
$0$
};

\node[inner sep=0pt,align=center, scale=0.8] (time) at (5.75,-0.3)
{
$1$
};

\draw[-, color=black] (5.7,0.1) to (5.7,-0.1);

\draw[-, color=black] (0,0.1) to (0,-0.1);

\draw [decorate,decoration={brace,amplitude=5pt,mirror,raise=4pt},yshift=0pt,rotate=180, line width=0.20mm]
(-5.65,0.1) -- (-4.35,0.1) node [black,midway,xshift=0.1cm] {};

\draw [decorate,decoration={brace,amplitude=5pt,mirror,raise=4pt},yshift=0pt,rotate=180, line width=0.20mm]
(-5.65,-0.5) -- (-3.2,-0.5) node [black,midway,xshift=0.1cm] {};

\node[inner sep=0pt,align=center, scale=0.8] (time) at (5,0.45)
{
$\mathscr{S}^{(\mathrm{A})}_{1,1,\pi_{\mathrm{D}}}$
};
\node[inner sep=0pt,align=center, scale=0.8] (time) at (4.4,1.05)
{
$\mathscr{S}^{(\mathrm{A})}_{1,L,\pi_{\mathrm{D}}}$
};


%

\draw[-, color=black] (4.3,0.1) to (4.3,-0.1);
\draw[-, color=black] (3.25,0.1) to (3.25,-0.1);


\node[inner sep=0pt,align=center, scale=0.8] (time) at (4.3,-0.3)
{
$\tilde{\beta}_{1,1}$
};
\node[inner sep=0pt,align=center, scale=0.8] (time) at (3.28,-0.3)
{
$\tilde{\beta}_{1,L}$
};


\draw[-, color=black] (1.15,0.1) to (1.15,-0.1);

\node[inner sep=0pt,align=center, scale=0.8] (time) at (1.1,-0.3)
{
$\tilde{\beta}_{0,1}$
};

\node[inner sep=0pt,align=center, scale=0.8] (time) at (1.725,-0.3)
{
$\hdots$
};
\draw[-, color=black] (2.35,0.1) to (2.35,-0.1);
\node[inner sep=0pt,align=center, scale=0.8] (time) at (2.3,-0.3)
{
$\tilde{\beta}_{0,L}$
};

\node[inner sep=0pt,align=center, scale=0.8] (time) at (3.8,-0.3)
{
$\hdots$
};

\draw [decorate,decoration={brace,amplitude=5pt,mirror,raise=4pt},yshift=0pt,rotate=180, line width=0.20mm]
(-2.3,-0.5) -- (-0.05,-0.5) node [black,midway,xshift=0.1cm] {};

\draw [decorate,decoration={brace,amplitude=5pt,mirror,raise=4pt},yshift=0pt,rotate=180, line width=0.20mm]
(-1.1,0.1) -- (-0.05,0.1) node [black,midway,xshift=0.1cm] {};

\node[inner sep=0pt,align=center, scale=0.8] (time) at (0.55,0.45)
{
$\mathscr{S}^{(\mathrm{A})}_{0,1,\pi_{\mathrm{D}}}$
};
\node[inner sep=0pt,align=center, scale=0.8] (time) at (1.1,1.05)
{
$\mathscr{S}^{(\mathrm{A})}_{0,L,\pi_{\mathrm{D}}}$
};
\end{tikzpicture}
};

\end{tikzpicture}
  }
  \caption{Illustration of Theorem \ref{thm:best_responses}; the upper part shows $L$ thresholds $\tilde{\alpha}_{1} \geq \tilde{\alpha}_{2}, \hdots, \geq \tilde{\alpha}_{L} \in [0,1]$ that define a best response strategy $\tilde{\pi}_{\mathrm{D}} \in \mathscr{B}_{\mathrm{D}}(\pi_{\mathrm{A}})$ for the defender (\ref{eq:prop_br_defender}); the lower part shows $2L$ thresholds $\tilde{\beta}_{0,1}, \tilde{\beta}_{1,1}, \hdots, \tilde{\beta}_{0,L}, \tilde{\beta}_{1,L} \in [0,1]$ that define a best response strategy $\tilde{\pi}_{\mathrm{A}} \in \mathscr{B}_{\mathrm{A}}(\pi_{\mathrm{D}})$ for the attacker (\ref{eq:prop_br_attacker_1})--(\ref{eq:prop_br_attacker_2}).}
    \label{fig:threshold_policy_3}
  \end{figure}
Theorem \ref{thm:best_responses} tells us that $\Gamma$ has a mixed Nash equilibrium. Further, under assumptions generally met in practice, the best response strategies have threshold properties (see Fig. \ref{fig:threshold_policy_3}). In the following, we describe an algorithm that leverages these properties to efficiently approximate Nash equilibria of $\Gamma$.

\subsection{Finding Nash Equilibria through Fictitious Self-Play}\label{sec:rl_approach}
Computing Nash equilibria for a \posg is generally intractable \cite[Thm. 3.5]{NIPS2007_3435c378} \cite[Thm. 6]{pspace_complexity}. However, approximate solutions can be obtained through iterative methods. One such method is \textit{fictitious self-play}, where both players start from random strategies and continuously update their strategies based on outcomes of played game episodes \cite{brown_fictious_play}.

Fictitious self-play evolves through a sequence of iteration steps, which is illustrated in Fig. \ref{fig:fp_2}. An iteration step includes three stages. First, player $1$ learns a best response strategy against player $2$'s current strategy. The roles are then reversed and player $2$ learns a best response strategy against player $1$'s current strategy. Lastly, each player adopts a new strategy, which is determined by the empirical distribution over its past best response strategies. The sequence of iteration steps continues until the strategies of both players have sufficiently converged to a Nash equilibrium \cite[Thms. 7.2.4-7.2.5]{multiagent_systems_book_1}.
\subsection{Our Self-Play Algorithm: \tfp}\label{sec:t_fp}
We present a fictitious self-play algorithm called Threshold Fictitious Self-Play (\tfp), which efficiently approximates a Nash equilibrium of $\Gamma$ based on Theorem \ref{thm:best_responses}. The pseudocode of \tfp is listed in Algorithm \ref{alg:ne_approximation}.

\begin{figure}
  \centering
  \scalebox{1.25}{
    \input{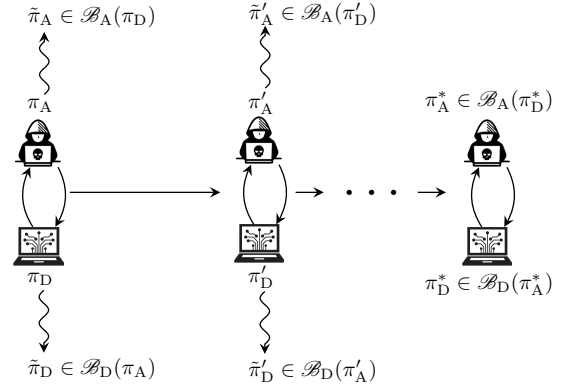}
  }
  \caption{The fictitious self-play process; in every iteration step each player learns a best response strategy $\tilde{\pi}_{i} \in \mathscr{B}_i(\pi_{{-i}})$ and updates its strategy based on the empirical distribution of its past best response strategies; the horizontal arrows indicate iteration steps of self-play and the vertical arrows indicate the learning of best response strategies; the process converges towards a Nash equilibrium $(\pi_{\mathrm{D}}^{*},\pi_{\mathrm{A}}^{*})$.} \label{fig:fp_2}
\end{figure}
\tfp implements the fictitious self-play process described above and generates a sequence of strategy profiles $(\pi_{\mathrm{D}}, \pi_{\mathrm{A}})$, $(\pi^{\prime}_{\mathrm{D}}$, $\pi^{\prime}_{\mathrm{A}})$, $\hdots$ that converges to a Nash equilibrium $(\pi^{*}_{\mathrm{D}}, \pi^{*}_{\mathrm{A}})$ \cite[Thms. 7.2.4-7.2.5]{multiagent_systems_book_1}. During each step of this process, \tfp learns best responses against the players' current strategies and then updates the strategies of both players (see Fig. \ref{fig:fp_2}).

To learn the best response strategies $\tilde{\pi}_{\mathrm{D}}\in \mathscr{B}_{\mathrm{D}}(\pi_{\mathrm{A}})$ and $\tilde{\pi}_{\mathrm{A}} \in \mathscr{B}_{\mathrm{A}}(\pi_{\mathrm{D}})$, \tfp parameterizes $\tilde{\pi}_{\mathrm{D}}$ and $\tilde{\pi}_{\mathrm{A}}$ through threshold vectors according to Theorem \ref{thm:best_responses}. The defender's best response strategy $\tilde{\pi}_{\mathrm{D}}$ is parameterized with the vector $\tilde{\bm{\theta}}^{(\mathrm{D})} \in \mathbb{R}^{L}$ (\ref{eq:smooth_threshold}). Similarly, the attacker's best response strategy $\tilde{\pi}_{\mathrm{A}}$ is parameterized with the vector $\tilde{\bm{\theta}}^{(\mathrm{A})} \in \mathbb{R}^{2L}$ (\ref{eq:smooth_threshold_2}).
\begin{align}
\varphi(a,b) &\triangleq \left(1 + \left(\frac{b(1-\sigma(a))}{\sigma(a)(1-b)}\right)^{-20}\right)^{-1} \label{eq:phi_eq}\\
\tilde{\pi}_{\mathrm{D},\tilde{\bm{\theta}}^{(\mathrm{D})}}\big(\mathfrak{S} \mid b(1)\big) &\triangleq \varphi\left(\tilde{\bm{\theta}}^{(\mathrm{D})}_l, b(1)\right) \label{eq:smooth_threshold}\\
\tilde{\pi}_{\mathrm{A},\tilde{\bm{\theta}}^{(\mathrm{A})}}\big(\mathfrak{S}\mid b(1),s\big) &\triangleq \varphi\left(\tilde{\bm{\theta}}^{(\mathrm{A})}_{sL+l}, \pi_{\mathrm{D}}(\mathfrak{S} \mid b(1))\right)\label{eq:smooth_threshold_2}
\end{align}
The parameterized strategies defined by (\ref{eq:phi_eq})--(\ref{eq:smooth_threshold_2}) are mixed (and differentiable) strategies that approximate threshold strategies (see Fig. \ref{fig:smooth_threshold}). In (\ref{eq:phi_eq})--(\ref{eq:smooth_threshold_2}), $\sigma(\cdot)$ is the sigmoid function, $a \in \mathbb{R}$, and $b \in \mathbb{R}$. Further, $\sigma(\tilde{\bm{\theta}}^{(\mathrm{D})}_{1})$, $\sigma(\tilde{\bm{\theta}}^{(\mathrm{D})}_{2})$, $\hdots$, $\sigma(\tilde{\bm{\theta}}^{(\mathrm{D})}_{L}) \in [0,1]$ are the $L$ thresholds of the defender (see Theorem \ref{thm:best_responses}.B) and $\sigma(\tilde{\bm{\theta}}^{(\mathrm{A})}_{1})$, $\sigma(\tilde{\bm{\theta}}^{(\mathrm{A})}_{2})$, $\hdots$, $\sigma(\tilde{\bm{\theta}}^{(\mathrm{A})}_{2L}) \in [0,1]$ are the $2L$ thresholds of the attacker (see Theorem \ref{thm:best_responses}.C).

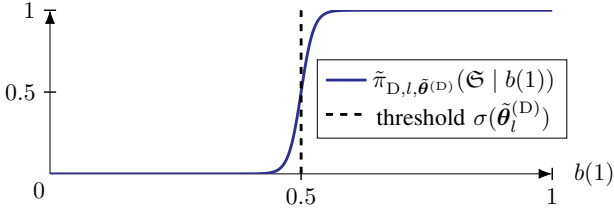
\begin{figure}
  \centering
  \scalebox{0.9}{
    \begin{tikzpicture}[
    dot/.style={
        draw=black,
        fill=blue!90,
        circle,
        minimum size=3pt,
        inner sep=0pt,
        solid,
    },
    ]
\tikzset{
        hatch distance/.store in=\hatchdistance,
        hatch distance=10pt,
        hatch thickness/.store in=\hatchthickness,
        hatch thickness=2pt
      }
\pgfdeclarepatternformonly[\hatchdistance,\hatchthickness]{flexible hatch}
    {\pgfqpoint{0pt}{0pt}}
    {\pgfqpoint{\hatchdistance}{\hatchdistance}}
    {\pgfpoint{\hatchdistance-1pt}{\hatchdistance-1pt}}%
    {
        \pgfsetcolor{\tikz@pattern@color}
        \pgfsetlinewidth{\hatchthickness}
        \pgfpathmoveto{\pgfqpoint{0pt}{0pt}}
        \pgfpathlineto{\pgfqpoint{\hatchdistance}{\hatchdistance}}
        \pgfusepath{stroke}
      }

\node[scale=1] (kth_cr) at (0,2.15)
{
  \begin{tikzpicture}[declare function={sigma(\x)=1/(1+exp(-\x));
      sigmap(\x)=sigma(\x)*(1-sigma(\x));}]
    \pgfmathsetlengthmacro\MajorTickLength{
      \pgfkeysvalueof{/pgfplots/major tick length} * 1.2
    }
    \begin{axis}[
        xmin=0,
        xmax=1,
        ymin=0,
        ymax=1,
        width = 9cm,
        height = 4cm,
        axis lines=center,
        xtick distance=0.5,
        ytick distance=0.5,
        x tick style={color=black},
        y tick style={color=black},
        major tick length=\MajorTickLength,
      every tick/.style={
        black,
        semithick,
      },
        xlabel style={below right},
        ylabel style={above left},
        axis line style={-{Latex[length=2mm]}},
        smooth,
        legend style={at={(1.03,0.7)}},
        legend columns=1,
        legend style={
            /tikz/column 2/.style={
                column sep=5pt,
              }
              }
        ]

\addplot[Blue,mark=none,samples=300,smooth, name path=l1, very thick, domain=0:1]   (x,{1/(1+ ( (x*(1-sigma(0)))/(sigma(0)*(1-x)) )^(-20))});
\addplot[black,dashed,mark=none,smooth, name path=l1, very thick,domain=0:1]   (0.5,x);

\legend{$\tilde{\pi}_{\mathrm{D},l,\tilde{\bm{\theta}}^{(\mathrm{D})}}(\mathfrak{S} \mid b(1))$, threshold $\sigma(\tilde{\bm{\theta}}^{(\mathrm{D})}_l)$}
\end{axis}
\node[inner sep=0pt,align=center, scale=1, rotate=0, opacity=1] (obs) at (8.1,0)
{
  $b(1)$
};

\node[inner sep=0pt,align=center, scale=1, rotate=0, opacity=1] (obs) at (-0.1,-0.25)
{
  $0$
};
\end{tikzpicture}
};

  \end{tikzpicture}
 }
 \caption{A mixed threshold strategy where $\sigma(\tilde{\bm{\theta}}_l^{(\mathrm{D})})$ is the threshold ($0.5$ in this example); the x-axis indicates the defender's belief state $b(1) \in [0,1]$ and the y-axis indicates the probability prescribed by $\tilde{\pi}_{\mathrm{D},\tilde{\bm{\theta}}^{(\mathrm{D})}}$ to the stop action $\mathfrak{S}$.}
    \label{fig:smooth_threshold}
  \end{figure}
  Using this parameterization, \tfp learns the best response strategies $\tilde{\pi}_{\mathrm{D},\tilde{\bm{\theta}}^{(\mathrm{D})}}$ and $\tilde{\pi}_{\mathrm{A},\tilde{\bm{\theta}}^{(\mathrm{A})}}$ by iteratively updating the threshold vectors $\tilde{\bm{\theta}}^{(\mathrm{D})}$ and $\tilde{\bm{\theta}}^{(\mathrm{A})}$ through stochastic approximation. To update the threshold vectors, \tfp simulates $\Gamma$, which allows to evaluate the objective functions $J_{\mathrm{D}}(\tilde{\pi}_{\mathrm{D},\tilde{\bm{\theta}}^{(\mathrm{D})}}, \pi_{\mathrm{A}})$ and $J_{\mathrm{A}}(\pi_{\mathrm{D}}, \tilde{\pi}_{\mathrm{A},\tilde{\bm{\theta}}^{(\mathrm{A})}})$ (\ref{eq:objective_1})--(\ref{eq:objective_2}). The obtained values of $J_{\mathrm{D}}$ and $J_{\mathrm{A}}$ are then used to estimate the gradients $\nabla_{\tilde{\bm{\theta}}^{(\mathrm{D})}}J_{\mathrm{D}}$ and $\nabla_{\tilde{\bm{\theta}}^{(\mathrm{A})}}J_{\mathrm{A}}$ using the Simultaneous Perturbation Stochastic Approximation (\spsa) gradient estimator (lines 10-19 in Algorithm \ref{alg:ne_approximation}) \cite[\S 2]{spsa} \cite[\S 3]{spsa_impl}. Next, the estimated gradients are used to update $\tilde{\bm{\theta}}^{(\mathrm{D})}$ and $\tilde{\bm{\theta}}^{(\mathrm{A})}$ through stochastic gradient ascent (line 20).

  The above procedure of estimating gradients and updating $\tilde{\bm{\theta}}^{(\mathrm{D})}$ and $\tilde{\bm{\theta}}^{(\mathrm{A})}$ continues for a given number of iterations (lines 9-21). After these iterations have finished, the threshold vectors $\tilde{\bm{\theta}}^{(\mathrm{D})}$ and $\tilde{\bm{\theta}}^{(\mathrm{A})}$ are added to buffers $\Theta^{(\mathrm{D})}$ and $\Theta^{(\mathrm{A})}$, which contain the vectors learned in previous iterations of \tfp (line 22). Finally, the \tfp iteration step is completed by having both players update their strategies based on the empirical distributions over the past vectors in the buffers (lines 24-25).

The sequence of iteration steps described above continues until the strategies have sufficiently converged to a Nash equilibrium (lines 6-27). (In Algorithm \ref{alg:ne_approximation}, $\mathcal{U}_{k}(\{-1,1\})$ denotes a $k$-dimensional discrete multivariate uniform distribution on $\{-1,1\}$ and $\pi_{-i}$ denotes the strategy of player $j\in \mathcal{N}\setminus \{i\}$.)

\begin{algorithm}
  \caption{\tfp: Threshold Fictitious Self-Play}\label{alg:ne_approximation}
  \hspace*{\algorithmicindent} \textbf{Input} \\
  \hspace*{\algorithmicindent}  $\Gamma, N$: the \posg and $\#$ best response iterations\\
  \hspace*{\algorithmicindent}  $a,c,\lambda,A,\epsilon,\delta$: scalar coefficients\\
  \hspace*{\algorithmicindent} \textbf{Output} \\
  \hspace*{\algorithmicindent}  $(\pi^{*}_{\mathrm{D}}, \pi^{*}_{\mathrm{A}})$: an approximate Nash equilibrium
\begin{algorithmic}[1]
  \Procedure{t-fp}{}
  \State $\tilde{\bm{\theta}}^{(\mathrm{D})} \sim \mathcal{U}_L(\{-1,1\})$, $\quad \tilde{\bm{\theta}}^{(\mathrm{A})} \sim \mathcal{U}_{2L}(\{-1,1\})$
  \State $\Theta^{(\mathrm{D})} \leftarrow \{\tilde{\bm{\theta}}^{(\mathrm{D})}\}, \quad \Theta^{(\mathrm{A})} \leftarrow \{\tilde{\bm{\theta}}^{(\mathrm{A})}\}, \quad \hat{\delta} \leftarrow \infty$
  \State $\pi_{\mathrm{D}}\leftarrow \text{\textsc{EmpiricalDistribution}($\Theta^{(\mathrm{D})}$)}$
  \State $\pi_{\mathrm{A}}\leftarrow \text{\textsc{EmpiricalDistribution}($\Theta^{(\mathrm{A})}$)}$
  \While{$\hat{\delta} \geq \delta$}
  \For{$i \in \{\mathrm{D}, \mathrm{A}\} = \{1,2\}$}
  \State $\tilde{\bm{\theta}}_{(1)}^{(i)} \sim \mathcal{U}_{iL}(\{-1,1\})$
  \For{$n \in \{1, \hdots, N\}$}
  \State $a_n \leftarrow \frac{a}{(n + A)^{\epsilon}}, \quad c_n \leftarrow \frac{c}{n^{\lambda}}$
  \For{$k \in \{1, \hdots, iL\}$}
  \State $(\Delta_n)_k \sim \mathcal{U}_1(\{-1,1\})$
  \EndFor
  \State $R_{high} \sim J_i(\pi_{i,\tilde{\bm{\theta}}^{(i)}_{(n)}} + c_n\Delta_n, \pi_{-i})$
  \State $R_{low} \sim J_i(\pi_{i,\tilde{\bm{\theta}}^{(i)}_{(n)}} - c_n\Delta_n, \pi_{-i})$
  \For{$k \in \{1, \hdots, iL\}$}
  \State $G \leftarrow \frac{R_{high} - R_{low}}{2c_n(\Delta_n)_{k}}$
  \State $\left(\hat{\nabla}_{\tilde{\bm{\theta}}^{(i)}_{(n)}}J_i(\pi_{i,\tilde{\bm{\theta}}^{(i)}_{(n)}}, \pi_{-i})\right)_{k} \leftarrow G$
  \EndFor
  \State $\tilde{\bm{\theta}}^{(i)}_{(n+1)} = \tilde{\bm{\theta}}^{(i)}_{(n)} + a_n\hat{\nabla}_{\tilde{\bm{\theta}}^{(i)}_{(n)}}J_i(\pi_{i,\tilde{\bm{\theta}}^{(i)}_{(n)}},\pi_{-i})$
  \EndFor
  \State $\Theta^{(i)} \leftarrow \Theta^{(i)} \cup \tilde{\bm{\theta}}^{(i)}_{(N+1)}$
  \EndFor
  \State $\pi_{\mathrm{D}}\leftarrow \text{\textsc{EmpiricalDistribution}($\Theta^{(\mathrm{D})}$)}$
  \State $\pi_{\mathrm{A}}\leftarrow \text{\textsc{EmpiricalDistribution}($\Theta^{(\mathrm{A})}$)}$
  \State $\hat{\delta} = \text{\textsc{Exploitability}}(\pi_{\mathrm{D}}, \pi_{\mathrm{A}})$
  \EndWhile
  \State \Return $(\pi_{\mathrm{D}}, \pi_{\mathrm{A}})$
\EndProcedure
\end{algorithmic}
\end{algorithm}
\section{Emulating the Target Infrastructure\\ to Instantiate the Simulation \\and to Evaluate Learned Strategies}\label{sec:policy_learning_results}
The \tfp algorithm described above approximates a Nash equilibrium of $\Gamma$ by simulating game episodes and updating both players' strategies through stochastic approximation. \tfp requires the observation distribution conditioned on the system state $f_{O\mid s}$ (\ref{eq:obs_1})--(\ref{eq:obs_2}). The emulation system shown in Fig. \ref{fig:method} allows us to estimate this distribution and later to evaluate the learned strategies.

This section describes the emulation system, our method for estimating $f_{O\mid s}$, and our method for evaluating defender strategies.
\subsection{Emulating the Target Infrastructure}\label{sec:emu_target_inf}
The emulation system executes on a cluster of machines that runs a virtualization layer provided by Docker containers and virtual links \cite{docker}. The system implements network isolation and traffic shaping using network namespaces and the \netem module in the Linux kernel \cite{netem}. Resource allocation to containers, e.g. \cpu and memory, is enforced using \cgroups.

The network topology of the emulated infrastructure is shown in Fig. \ref{fig:system2} and its configuration is given in Appendix \ref{appendix:infrastructure_configuration}. The emulation system includes the clients, the attacker, the defender, network connectivity, and $31$ devices of the target infrastructure (e.g. application servers and the gateway). The software functions on the emulation system replicate important components of the target infrastructure, such as, web servers, databases, and the Snort \idps, which is deployed using Snort's community ruleset v2.9.17.1.

We emulate connections between servers as full-duplex loss-less connections of $1$ Gbit/s capacity in both directions. We emulate connections between the gateway and the external client population as full-duplex connections of $100$ Mbit/s capacity and $0.1\%$ packet loss with random bursts of $1\%$ packet loss. (These numbers are based on measurements on enterprise and wide-area networks \cite{packet_losses_decreasing,Paxson97end-to-endinternet,elliott_markov_chain_ref}.)

Technical documentation and application programming interfaces (\textsc{api}s) of the emulation system are available in \cite{csle_docs}.
\subsection{Emulating the Client Population}\label{sec:emu_target_inf}
The \textit{client population} is emulated by processes in Docker containers. Clients interact with application servers through the gateway by performing a sequence of functions on a sequence of servers, both of which are selected uniformly at random from Table \ref{tab:client_profiles}. Client arrivals per time-step are emulated using a stationary Poisson process with mean $\lambda=20$ and exponentially distributed service times with mean $\mu=4$. The duration of a time-step is $30$ seconds.
\begin{table}
\centering
\begin{tabular}{ll} \toprule
  {\textit{Functions}} & {\textit{Application servers}} \\ \midrule
  \http, \ssh, \snmp, \icmp & $N_2,N_3,N_{10},N_{12}$\\
  \irc, \postgres, \snmp & $N_{31},N_{13},N_{14},N_{15},N_{16}$\\
  \ftp, \dns, \telnet & $N_{10}, N_{22}, N_{4}$ \\
  \bottomrule\\
\end{tabular}
\caption{Emulated client population; each client invokes functions on application servers.}\label{tab:client_profiles}
\end{table}
\subsection{Emulating Defender and Attacker Actions}\label{sec:emu_player_actions}
The defender and the attacker observe the infrastructure continuously and take actions at time-steps $t=1,2,\hdots, T$. During each step, the defender and the attacker perform one action each.

The defender executes either a continue action or a stop action. A continue action is virtual in the sense that it does not trigger any function in the emulation. A stop action, however, invokes specific functions in the emulated infrastructure. We have implemented $L=7$ stop actions for the defender, which are listed in Table \ref{tab:defender_stop_actions}. The first stop action revokes user certificates and recovers user accounts expected to be compromised by the attacker. The second stop action updates the firewall configuration of the gateway to drop traffic from IP addresses flagged by the \idps. Stop actions $3$--$6$ trigger the dropping of traffic that generates \idps alerts of priorities $1$--$4$. The final stop action blocks all incoming traffic. (Note that according to Snort's terminology, $1$ is the highest priority. We inverse the labeling in our framework for convenience.)

\begin{table}
  \centering
\resizebox{1\columnwidth}{!}{%
\begin{tabular}{ll} \toprule
  {\textit{Stop index}} & {\textit{Action}} \\ \midrule
  $1$ & Revoke user certificates \\
  $2$ & Blacklist IPs \\
  $3$ & Drop traffic that generates \idps alerts of priority $1$ \\
  $4$ & Drop traffic that generates \idps alerts of priority $2$ \\
  $5$ & Drop traffic that generates \idps alerts of priority $3$ \\
  $6$ & Drop traffic that generates \idps alerts of priority $4$ \\
  $7$ & Block gateway \\
  \bottomrule\\
\end{tabular}
}
\caption{Defender commands executed on the emulation system.}\label{tab:defender_stop_actions}
\end{table}
\begin{table}
\centering
\begin{tabular}{ll} \toprule
  {\textit{Type}} & {\textit{Actions}} \\ \midrule
  Reconnaissance  & \tcpp \syn scan, \udp port scan, \\
                  & \tcpp \nulll scan, \tcpp \xmas scan, \tcpp \finn scan, \\
                  & ping scan, \tcpp connection scan, \vulscan \\
  &\\
  Brute-force attack & \telnet, \ssh, \ftp, \cassandra,\\
                  &  \irc, \mongo, \mysql, \smtp, \postgres\\
                  &\\
  Exploit & \cve-2017-7494, \cve-2015-3306,\\
                  & \cve-2010-0426, \cve-2015-5602, \\
                  &  \cve-2014-6271, \cve-2016-10033\\
                  & \cve-2015-1427, \cwe-89\\
  \bottomrule\\
\end{tabular}
\caption{Attacker commands executed on the emulation system; exploits are identified according to their corresponding vulnerability and its identifier in the Common Vulnerabilities and Exposures (\cve) database \cite{cve} and in the Common Weakness Enumeration (\cwe) list \cite{cwe}.}\label{tab:attacker_actions}
\end{table}

Like the defender, the attacker executes either a stop action or a continue action during each time-step. The attacker can only take two stop actions during a game episode. The first determines when the intrusion starts and the second when it terminates (see \S \ref{sec:formal_model_2}).

During an intrusion, the attacker executes a sequence of commands, drawn randomly from all of the commands listed in Table \ref{tab:attacker_actions}. (Detailed descriptions of all commands are available in Appendix \ref{appendix:attacker actions}). The first command in this sequence is executed when the attacker takes the first stop action. A further command is invoked whenever the attacker takes a continue action.

\subsection{Estimating the \idps Alert Distribution}\label{sec:estimating_dist}
At the end of every time-step, the emulation system collects the number of \idps alerts with priorities $1$--$4$ that occurred during the time-step. These values are then used to compute the metric $o_t$, which contains the total number of \idps alerts, weighted by priority.
\begin{figure}
  \centering
    \scalebox{0.83}{
      \includegraphics{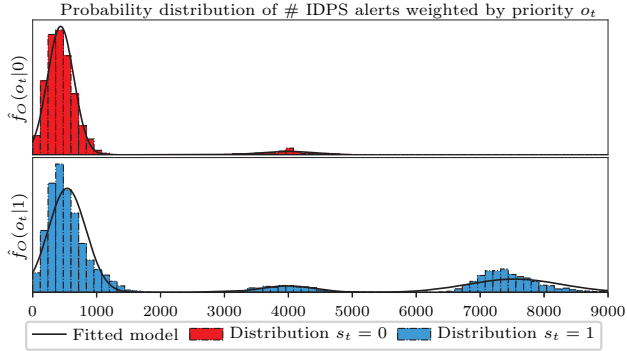}
    }
    \caption{Empirical distributions of $o_t$ when no intrusion occurs ($s_t=0$) and during intrusion ($s_t=1$); the black lines show the fitted Gaussian mixture models.}
    \label{fig:ids_distribution}
\end{figure}

For the evaluation reported in this paper we collect measurements from $23,000$ time-steps. Using these measurements, we apply expectation-maximization \cite{em_demp_77} to fit Gaussian mixture distributions $\hat{f}_{O\mid 0}$ and $\hat{f}_{O\mid 1}$ as estimates of $f_{O\mid 0}$ and $f_{O\mid 1}$ (\ref{eq:obs_1})--(\ref{eq:obs_2}).

Fig. \ref{fig:ids_distribution} shows the empirical distributions and the fitted models over the discrete observation space $\mathcal{O} = \{1,2,\hdots,9000\}$. $\hat{f}_{O \mid 0}$ and $\hat{f}_{O \mid 1}$ are Gaussian mixtures with two and three components, respectively. Both mixtures have most probability mass within $0$--$1000$. $\hat{f}_{O \mid 1}$ also has substantial probability mass at larger values.

The stochastic matrix with the rows $\hat{f}_{O \mid 0}$ and $\hat{f}_{O \mid 1}$ has about $72 \times 10^{6}$ second-order minors, which are almost all non-negative. This suggests to us that the \tpp assumption in Theorem \ref{thm:best_responses} can be made.
\subsection{Running a Game Episode}\label{sec:simulation_episode}
During a game episode, the state evolves according to the dynamics defined by (\ref{eq:tp_1})--(\ref{eq:tp_7}), the defender's belief state evolves according to (\ref{eq:belief_upd}), the players' rewards are calculated using the reward function $\mathcal{R}$ (\ref{eq:reward_0})--(\ref{eq:reward_5}), the defender's observations are obtained from $f_O$ (\ref{eq:obs_1})--(\ref{eq:obs_2}), and the actions of both players are determined by their respective strategies. If the game runs in the emulation system, the players' actions include executing networking and computing functions (see Tables \ref{tab:defender_stop_actions}--\ref{tab:attacker_actions}), and the observations from $f_O$ are obtained through reading log files and metrics of the emulated infrastructure. (To collect the logs and system metrics from the emulation, we run software sensors that write to a distributed queue implemented with Kafka \cite{kafka}.) In the case of a game in the simulation system, the observations are instead sampled from the estimated distribution $\hat{f}_O$.
\section{Learning Nash Equilibrium Strategies for the Target Infrastructure}\label{sec:eval}
Our approach to finding near-optimal defender strategies includes: (\textit{i}) emulating the target infrastructure to obtain statistics for instantiating the simulation system; (\textit{ii}) learning Nash equilibrium strategies using the \tfp algorithm in \S \ref{sec:game_analysis}; and (\textit{iii}) evaluating learned strategies on the emulation system in \S \ref{sec:policy_learning_results} (see Fig. \ref{fig:method}). This section describes the learning process and the evaluation results of the intrusion response use case.

\subsection{Learning Equilibrium Strategies through Self-Play}
\begin{figure*}
\centering
\scalebox{0.94}{
      \includegraphics{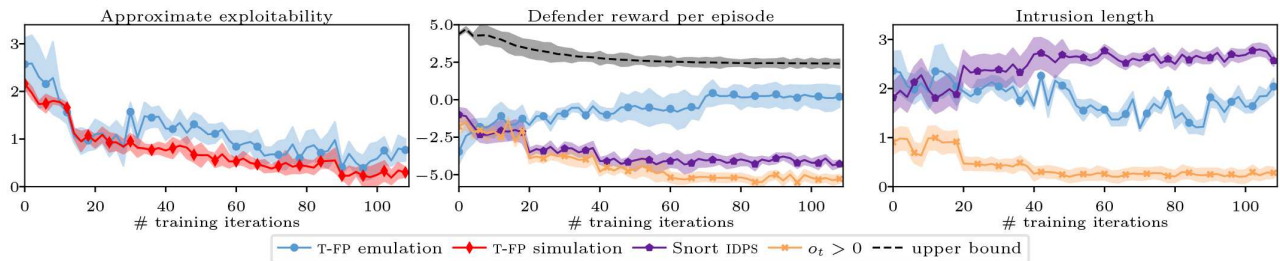}
}
\caption{Learning curves from the self-play process with \tfp; the red curve shows simulation results and the blue curves show emulation results; the purple, orange, and black curves relate to baseline strategies; the figures show different performance metrics: exploitability (\ref{eq:approx_exp}), episodic reward, and the length of intrusion; the curves indicate the mean and the $95\%$ confidence interval over four training runs with different random seeds.}
    \label{fig:exploitability_curve}
  \end{figure*}
We run \tfp for $500$ iteration steps to estimate a Nash equilibrium using the iterative method in \S \ref{sec:rl_approach}, which is sufficient to meet the termination condition (line 6 in Algorithm \ref{alg:ne_approximation}). These iteration steps generate a sequence of strategy pairs $(\pi_{\mathrm{D}}, \pi_{\mathrm{A}})_1,\allowbreak (\pi_{\mathrm{D}}, \pi_{\mathrm{A}})_2,\allowbreak \hdots,\allowbreak (\pi_{\mathrm{D}}, \pi_{\mathrm{A}})_{500}$.

At the end of each iteration step, we evaluate the current strategy pair $(\pi_{\mathrm{D}}, \pi_{\mathrm{A}})$ by running $500$ evaluation episodes in the simulation system and $5$ evaluation episodes in the emulation system. This process allows us to produce learning curves for different performance metrics (see Fig. \ref{fig:exploitability_curve}).

The $500$ training iterations and the associated evaluations constitute one \textit{training run}. We run four training runs with different random seeds. A single training run takes about $5$ hours of processing time in the simulation system. In addition, it takes around $12$ hours to evaluate the strategies on the emulation system. The hyperparameters of \tfp are listed in Appendix \ref{appendix:hyperparameters}.

\textbf{Computing environment for simulation and emulation.} The environment for running simulations and training strategies is a \textsc{tesla} \textsc{p}\small 100 \normalsize\textsc{gpu}.

The emulated infrastructure is deployed on a server with a $24$-core \textsc{intel} \textsc{xeon} \textsc{gold} \small $2.10$ GHz \normalsize \textsc{cpu} and $768$ \textsc{gb} \textsc{ram}. Documentation of the emulation system is available in \cite{csle_docs}.

The code for the simulation system and the measurement traces for the intrusion response use case are available at \cite{github_cnsm_21_hammar_stadler}. They can be used to validate our results and to extend this research.

\textbf{Convergence metric for \tfp.} To estimate the convergence of the sequence of strategy pairs generated by \tfp, we use the \textit{approximate exploitability} metric $\hat{\delta}$ \cite{approx_br}:
\begin{align}
\hat{\delta} = J_{\mathrm{D}}(\hat{\pi}_{\mathrm{D}}, \pi_{\mathrm{A}}) + J_{\mathrm{A}}(\pi_{\mathrm{D}}, \hat{\pi}_{\mathrm{A}}) \label{eq:approx_exp}
\end{align}
where $\hat{\pi}_{i}$ denotes an approximate best response strategy for player $i$ and the objective functions $J_{\mathrm{D}}$ and $J_{\mathrm{A}}$ are defined in (\ref{eq:objective_1}) and (\ref{eq:objective_2}), respectively. The closer $\hat{\delta}$ becomes to $0$, the closer $(\pi_{\mathrm{D}},\pi_{\mathrm{A}})$ is to a Nash equilibrium.

\textbf{Baseline algorithms.} We compare the performance of \tfp with that of two popular algorithms in previous work that use reinforcement learning and study use cases similar to ours \cite{nfsp_jamming_1_sim,nfsp_security_2,horak_bosansky_hsvi,horak_solving_one_sided_posgs,posg_sequetial_attacks_bosansky}. The first algorithm is Neural Fictitious Self-Play (\nfsp) \cite{heinrich_1}, which is a general fictitious self-play algorithm that does not exploit the threshold structures expressed in Theorem \ref{thm:best_responses}. The second algorithm is Heuristic Search Value Iteration (\hsvi) for one-sided \posgs \cite{horak_bosansky_hsvi}, which is a state-of-the-art dynamic programming algorithm for one-sided \posgs.

\textbf{Defender baseline strategies.} We compare the dynamic defender strategies learned through \tfp with three static baseline strategies. The first baseline prescribes the stop action when an \idps alert occurs, i.e., when $o_t > 0$. The second baseline is derived from the Snort \idps, which is a de-facto industry standard and can be considered state-of-the-art for our use case. This baseline uses the Snort \idps's recommendation system and takes a stop action when Snort has dropped $100$ IP packets (see Appendix \ref{appendix:infrastructure_configuration} for the Snort configuration). The third baseline assumes prior knowledge of the intrusion time and performs all $L$ stops during the $L$ subsequent time-steps.

Although a growing body of work uses reinforcement learning and game theory to find intrusion response strategies (see \S \ref{sec:related_work} for a review of the related work), a direct comparison between the defender strategies learned in our framework and those found in previous work is not feasible for two reasons. First, nearly all of the prior works have developed defender strategies for custom simulations \cite{hammar_stadler_cnsm_20,hammar_stadler_cnsm_21, elderman, schwartz_2020, oslo_pentest_rl, kurt_rl, microsoft_red_teaming, ridley_ml_defense, rl_cyberdefense_heartbleed, deep_hierarchical_rl_pentest, pentest_rl_rohit, adaptive_cyber_defense_pomdp_rl,game_cyber_rl_sim,cmu_ppo_selfplay,al_shaer_book_ppo_simulation,mec_game_rl_q_learning_sim,nfsp_jamming_1_sim,9923774,9916301,7573127,7568529,miehling_attack_graph,Wang2021GameTheoreticAI,HUANG2020101660,Nguyen2018,8750848,9559403,9754705,9096400, dynamic_game_linan_zhu,general_sum_markov_games_for_strategic_detection_of_apt,honeypot_game,DBLP:journals/compsec/HorakBTKK19,game_theoretic_modeling_ofhoneypot_selection,serkan_gyorgy_game,posg_cyber_deception_network_epidemic,stocahstic_games_security_indep_nodes_nguyen_alpcan_basar, hammar_stadler_cnsm_20,optimal_thresholds_for_ids,a_game_theoretic_ids_control_sys_alpcan_basar,zhu_basar_dynamic_policy_ids_config,zhang2019,fog_computing_irs,altman_jamming_1} and there is no obvious way to map their solutions to an emulated environment like ours (see Fig. \ref{fig:system2} and Appendix \ref{appendix:infrastructure_configuration}). Second, the few prior works that study emulated infrastructures similar to ours either consider static attackers in fully observed environments \cite{muzero_sdn,atmos,sdn_rl_ddos,9328143,deep_air,YUNGAICELANAULA2022103444,sdn_zolo,Iannucci2021AnIR,Wang2020AnID} or focus on use cases that are different from the one considered in this article \cite{9328143,5270307}.
\subsection{Evaluating the Learned Strategies}\label{sec:one_stop_evaluation}
\begin{figure}
\centering
\scalebox{0.86}{
      \includegraphics{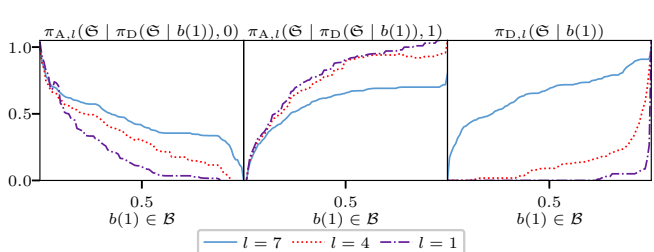}
}
\caption{Probability of the stop action $\mathfrak{S}$ by the learned equilibrium strategies in function of $b(1)$ and $l$; the left and middle plots show the attacker's stopping probability when $s=0$ and $s=1$, respectively; the right plot shows the defender's stopping probability.}
    \label{fig:stop_prob2}
  \end{figure}

\begin{figure}
\centering
\scalebox{0.82}{
      \includegraphics{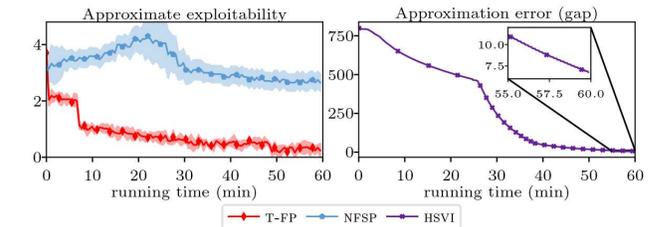}
}
\caption{Comparison between \tfp and two baseline algorithms: \nfsp and \hsvi; all curves show simulation results; the red curve relates to \tfp; the blue curve to \nfsp; the purple curve to \hsvi; the left plot shows the approximate exploitability metric (\ref{eq:approx_exp}) and the right plot shows the \hsvi approximation error \cite{horak_bosansky_hsvi}; the curves depicting \tfp and \nfsp show the mean and the $95\%$ confidence interval over four training runs with different random seeds.}
    \label{fig:converge_times}
  \end{figure}
Figure \ref{fig:exploitability_curve} shows the learning curves of the strategies obtained during the \tfp self-play process and the baselines introduced above. The red curve represents the results from the simulator; the blue curves show the results from the emulation system; the purple curves give the performance of the Snort \idps baseline; the orange curves relate to the baseline strategy that mandates a stop action when an \idps alert occurs; and the dashed black curve gives the performance of the baseline strategy that assumes prior knowledge of the intrusion time.

We note that all learning curves in Fig. \ref{fig:exploitability_curve} converge, which suggests that the learned strategies converge as well. (Fig. \ref{fig:exploitability_curve} only shows the first $120$ iterations of the $500$ iterations we performed, as the curves converge after $100$ iterations.) Specifically, we observe that the approximate exploitability (\ref{eq:approx_exp}) of the learned strategies converges to small values (left plot), which indicates that the learned strategies approximate a Nash equilibrium both in the simulator and in the emulation system. Further, we see from the plot in the middle that both baseline strategies show decreasing performance as the attacker updates its strategy. In contrast, the defender strategy learned through \tfp improves its performance over time. This shows the benefit of a game-theoretic approach where the defender strategy is optimized against a dynamic attacker. Lastly, we notice that the average intrusion length when the defender follows the learned defender strategy and the Snort \idps baseline strategy is $2$ and $3$, respectively (right plot). In comparison, the average intrusion length when the defender follows the baseline strategy $o_t> 0$ is close to $0$, which indicates that it tends to prescribe all stop actions before an intrusion occurs.

Figure \ref{fig:stop_prob2} represents the strategies learned through \tfp in a simple form. The y-axis shows the probability of a stop action and the x-axis shows the defender's belief $b(1) \in \mathcal{B}$ that an intrusion occurs. The strategies are clearly stochastic. This is consistent with Theorem \ref{thm:best_responses}.A, which predicts a mixed Nash equilibrium. Further, Theorem \ref{thm:best_responses}.B predicts that the defender's stopping probability is increasing with respect to $b(1)$ and decreasing with $l$, which is visible in the right plot. Similarly, Theorem \ref{thm:best_responses}.C predicts that the attacker's stopping probability decreases with the defender's stopping probability when $s=0$ and increases when $s=1$, which can be seen in the left and the middle plot.

Figure \ref{fig:converge_times} compares \tfp with the two baseline algorithms \nfsp and \hsvi on the simulator. \nfsp implements fictitious self-play and can thus be compared with \tfp with respect to approximate exploitability (\ref{eq:approx_exp}). We observe in the left plot that \tfp converges much faster than \nfsp. We explain the rapid convergence of \tfp by its design, which exploits structural properties of the stopping game.

The right plot shows that \hsvi reaches an \hsvi approximation error below $5$ within an hour of processing time. Based on the recent literature we anticipated a much longer processing time \cite{horak_thesis, horak_solving_one_sided_posgs}. This suggests to us that \tfp and \hsvi have similar convergence properties. A more detailed comparison between \tfp and \hsvi is hard to perform due to the different nature of the two algorithms.

Figure \ref{fig:value_fun} shows the estimated value function of the game $\hat{V_l^{*}}: \mathcal{B} \rightarrow \mathbb{R}$ (\ref{eq:bellman_posg_1}), where $\hat{V_l^{*}}(b(1))$ is the expected cumulative reward when the game starts in the belief state $b(1)$, the defender has $l$ stops remaining, and both players follow optimal (equilibrium) strategies.

We see in Fig. \ref{fig:value_fun} that $\hat{V_l^{*}}$ is piece-wise linear and convex, as expected from the theory of one-sided \posgs \cite{horak_thesis}. The figure indicates that $\hat{V_l^{*}}(b(1)) \leq 0$ for all $b(1) \in \mathcal{B}$ and that $\hat{V_l^{*}}(1)=0$ for all $l \in \{1,\hdots,L\}$. Further, we note that the value of $\hat{V_l^{*}}$ is minimal when $b(1)$ is around $0.25$ and that the values for $l=1$ and $l=7$ are very close.

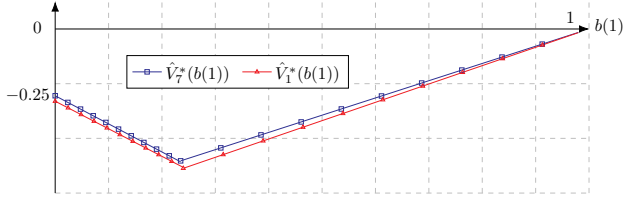
\begin{figure}
  \centering
  \scalebox{0.65}{
    \begin{tikzpicture}[
    dot/.style={
        draw=black,
        fill=blue!90,
        circle,
        minimum size=3pt,
        inner sep=0pt,
        solid,
    },
    ]

\node[scale=1] (kth_cr) at (0,2.15)
{
  \begin{tikzpicture}
    \begin{axis}[
      xmin=0,
      grid=major,
      grid style={dashed},
      xmax=1,
        ymin=-0.6,
        ymax=0.1,
        axis lines=center,
       ticks=none,
        xlabel style={below right},
        ylabel style={above left},
        axis line style={-{Latex[length=2mm]}},
        smooth,
        legend style={at={(0.55,0.72)}},
        legend columns=2,
        legend style={
            /tikz/column 2/.style={
                column sep=5pt,
              }
            },
            width=12.5cm,
            height=5.5cm
        ]


\addplot[Blue,mark=square,mark repeat=10,mark size=1.3pt,samples=100,smooth, name path=l1, domain=0:0.235]   (x,{(1-x)*-0.244144 + x*-1.27207});

\addplot[Red,mark=triangle,mark repeat=10,mark size=1.3pt,samples=100,smooth, name path=l2,  domain=0:0.24]   (x,{(1-x)*-0.264144 + x*-1.27207});

\addplot[Blue,mark=square,mark repeat=10,mark size=1.3pt,samples=100,smooth, name path=l1, domain=0.235:0.98]   (x,{(1-x)*-0.63033 + x*0});
\addplot[Red,mark=triangle,mark repeat=10,mark size=1.3pt,samples=100,smooth, name path=l2,  domain=0.24:0.98]   (x,{(1-x)*-0.67033 + x*0});

\legend{$\hat{V}^{*}_{7}(b(1))$, $\hat{V}^{*}_{1}(b(1))$}
\end{axis}
\node[inner sep=0pt,align=center, scale=1, rotate=0, opacity=1] (obs) at (10.6,3.6)
{
  $1$
};

\node[inner sep=0pt,align=center, scale=1, rotate=0, opacity=1] (obs) at (11.4,3.4)
{
  $b(1)$
};

\node[inner sep=0pt,align=center, scale=1, rotate=0, opacity=1] (obs) at (-0.3,3.4)
{
  $0$
};

\node[inner sep=0pt,align=center, scale=1, rotate=0, opacity=1] (obs) at (-0.5,2)
{
  $-0.25$
};

\end{tikzpicture}
};

  \end{tikzpicture}
 }
    \caption{The value function $\hat{V_l^{*}}(b(1))$ (\ref{eq:bellman_posg_1}) computed through the \hsvi algorithm with approximation error $4$; the blue and red curves relate to $l=7$ and $l=1$, respectively.}
    \label{fig:value_fun}
\end{figure}

That $\hat{V_l^{*}}(b(1)) \leq 0$ for all $b(1) \in \mathcal{B}$ and all $l \in \{1,\hdots,L\}$ has an intuitive explanation. For any $b(1)$, the attacker has the option to never attack if $s=0$ or to abort an attack if $s=1$. Both options yield a cumulative reward less than or equal to $0$ (\ref{eq:reward_0})--(\ref{eq:reward_5}). As a consequence, $\hat{V_l^{*}}(b(1)) \leq 0$ for any optimal attacker strategy and all $b(1) \in \mathcal{B}$ and $l \in \{1,\hdots,L\}$. (Recall that the attacker aims to minimize reward.)

The fact that $\hat{V_l^{*}}(b(1))=0$ when $b(1)=1$ can be understood as follows. $b(1)=1$ means that the defender knows that an intrusion occurs and will take defensive actions (see Theorem \ref{thm:best_responses}.B). Hence, when $b(1)=1$, the only way for the attacker to avoid detection is to abort the intrusion, which causes the game to end and yields a reward of zero, i.e. $\hat{V_l^{*}}(1)=0$ for all $l \in \{1,\hdots,L\}$.

We interpret the fact that $\argmin_{b(1)}\hat{V_l^{*}}(b(1))$ is around $0.25$ as follows. The value of $b(1)$ that obtains the minimum corresponds to the belief state where the attacker achieves the lowest expected reward in the game. Negative rewards in the game are obtained when the defender mistakes an intrusion for no intrusion and vice versa (\ref{eq:reward_0})--(\ref{eq:reward_5}). As a consequence, the attacker prefers belief states where the defender has a high uncertainty, e.g. $b(1)=0.5$. At the same time, the attacker does not want $b(1)$ to be so large that the defender performs all its defensive actions before it gets a chance to attack, which can explain why we find the minimum to be around $0.25$ rather than $0.5$.

Lastly, Fig. \ref{fig:client_data} shows the percentage of blocked attacker and client traffic when running repeated game episodes in the emulation system with different defender strategies. The x-axis shows the running time and the y-axis shows the percentage of blocked traffic per second.

We observe in the upper plot that all defender strategies block some client traffic, which is expected considering the false \idps alarms generated by the clients (see Fig. \ref{fig:ids_distribution}). (The defender actions that cause traffic to be dropped are listed in Table \ref{tab:defender_stop_actions}.) The $o_t>0$ baseline strategy blocks the most client traffic and the Snort \idps baseline strategy blocks the least, slightly less than the equilibrium strategy learned through \tfp.

We further observe in the lower plot that the equilibrium strategy learned by \tfp blocks the most attacker traffic and that the $o_t>0$ baseline strategy blocks the least. This suggests to us that the equilibrium strategy balances well the trade-off between blocking clients and the attacker based on the reward function (\ref{eq:reward_0})--(\ref{eq:reward_5}). In comparison, the $o_t > 0$ baseline implements a trivial defense strategy that blocks nearly all traffic, and the Snort \idps baseline blocks too little traffic, failing to stop the intrusion.
\begin{figure}
  \centering
    \scalebox{0.6}{
      \includegraphics{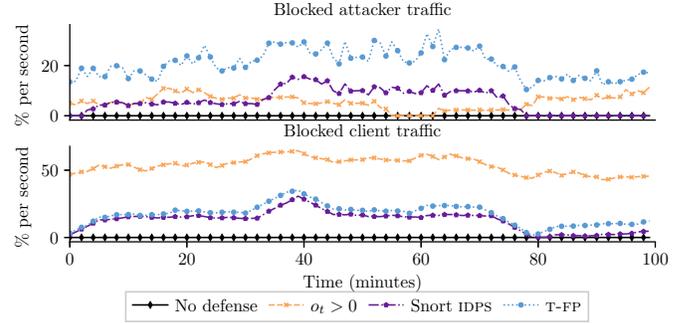}
    }
    \caption{Percentage of blocked attacker and client traffic in the emulation system; the blue curves show results from the equilibrium strategy learned via \tfp; the purple, orange, and black curves relate to baseline strategies.}
    \label{fig:client_data}
\end{figure}
\subsection{Discussion of the Evaluation Results}\label{sec:discussion}
In this work, we propose a framework for analyzing and solving the intrusion response use case, which we validate both theoretically and experimentally via simulation and emulation. The key findings can be summarized as follows:

(\textit{i}) Our framework is able to efficiently approximate optimal defender strategies for a practical IT infrastructure (see Fig. \ref{fig:exploitability_curve}). While we have not evaluated the learned strategies in the target infrastructure due to safety reasons, the fact that they achieve almost the same performance in the emulated infrastructure as in the simulator gives us confidence that the obtained strategies would perform as expected in the target infrastructure.

(\textit{ii}) The theory of optimal stopping provides insight about optimal strategies for attackers and defenders, which enables efficient computation of near-optimal strategies through self-play reinforcement learning (see Fig. \ref{fig:converge_times}). This finding can be explained by the threshold structures of the optimal stopping strategies, which drastically reduce the search space of possible strategies (see Theorem \ref{thm:best_responses} and Algorithm \ref{alg:ne_approximation}).

(\textit{iii}) The learned strategies can be efficiently implemented using the threshold properties. The computational complexity, which is dominated by the computation of the belief state, is upper bounded by $O(k|\mathcal{S}|^2|\mathcal{A}_{\mathrm{A}}|)$ where $k$ is a constant (\ref{eq:belief_upd}).

(\textit{iv}) Static defender strategies' performance deteriorate against a dynamic attacker, whereas defender strategies obtained through \tfp improve over time (see the middle plot in Fig. \ref{fig:exploitability_curve}). This finding is consistent with previous studies that use game-theoretic approaches (e.g. \cite{flipit,dynamic_game_linan_zhu}) and suggests limitations of static defense systems, such as the Snort \idps.
\section{Related Work}\label{sec:related_work}
Since the early 1990s, there has been a broad interest in automating network security functions, especially in the areas of intrusion detection, intrusion prevention, and intrusion response.

In the area of intrusion detection, the traditional approach has been to use packet inspection and static rules for detection of intrusions \cite{snort,ids_survey, int_prevention}. The main drawback of this approach lies in the need for domain experts to configure the rule sets. As a consequence, much effort has been devoted to developing statistical methods for detecting intrusions. Examples of statistical methods include anomaly detection methods (e.g. \cite{anomaly_ddetection}), change-point detection methods (e.g. \cite{tartakovsky_1}), Bayesian methods (e.g. \cite{fung_ids_distributed}), hidden Markov modeling methods (e.g. \cite{hmm_holgado}), deep learning methods (e.g. \cite{ml_anomaly_detection,deep_autoencoder_ids}), and threat intelligence methods (e.g. \cite{threat_intel_misp}). As a result of this effort, all mainstream \idss today have statistical components for automated detection of intrusions \cite{anderson_nides,snort,security_onion,wazuh,paxson1999bro}.

In contrast to intrusion detection, intrusion prevention and response usually remains a manual process performed by network administrators. Current \ipss and \irss can be configured with rules to automatically match response actions to known intrusion types, but they have no means to find effective response actions in an automatic way \cite{snort,trellix,932194,adepts_irs}. The problem of automatically finding response actions is an active area of research that uses concepts and methods from various fields, most notably from reinforcement learning (see surveys \cite{deep_rl_cyber_sec,control_rl_reviews,9272624} and textbook \cite{cybenko_acd}), control theory (see surveys \cite{dp_security_1,Miehling_control_theoretic_approaches_summary,feedback_control_computing_systems} and example \cite{Kreidl2004FeedbackCA}), causal modeling (see example \cite{causal_neil_agent}), game theory (see textbooks \cite{nework_security_alpcan,tambe,carol_book_intrusion_detection,levente_book}), graph theory (see examples \cite{acd_g_graph, purvine_graph_2}), fuzz testing (see examples \cite{fuzzing_rl_nddss_2021_wang,mayhem}), formal synthesis (see example \cite{formal_synthesis}), attack graphs (see example \cite{adepts_irs}), artificial intelligence (see surveys \cite{ai_survey,emanuello_math_cyber_defense} and textbook \cite{al_shaer_book}), and evolutionary methods (see examples \cite{armsrace_malware,hemberg_oreily_evo}).

While the research reported in this paper is informed by all the above works, we limit the following discussion to prior work that uses game-theoretic models and centers around finding security strategies through automatic control and reinforcement learning.
\subsection{Game-Theoretic Modeling in Network Security}
Since the early 2000s, researchers have studied automated security through modeling attacks and response actions on an IT infrastructure as a game between an attacker and a defender (see textbooks \cite{nework_security_alpcan,tambe,carol_book_intrusion_detection,levente_book}). The game is modeled in different ways depending on the use case. Examples from the literature include: advanced persistent threat games \cite{flipit,dynamic_game_linan_zhu,general_sum_markov_games_for_strategic_detection_of_apt,8750848,9559403,HUANG2020101660,apt_rl_simulation,HAN2022113912}, honeypot placement games \cite{honeypot_game,DBLP:journals/compsec/HorakBTKK19,game_theoretic_modeling_ofhoneypot_selection}, resource allocation games \cite{game_resource_alloc_malicious_packet,9923774}, authentication games \cite{serkan_gyorgy_game}, distributed denial-of-service games \cite{9328143,posg_cyber_deception_network_epidemic}, situational awareness games \cite{Brynielsson473862,Franke2014CyberSA}, moving target defense games \cite{sengupta_marl_stackel_bayesian,Nguyen2018}, jamming games \cite{altman_jamming_1}, and intrusion response games \cite{stocahstic_games_security_indep_nodes_nguyen_alpcan_basar, hammar_stadler_cnsm_20, muzero_sdn,optimal_thresholds_for_ids,a_game_theoretic_ids_control_sys_alpcan_basar,zhu_basar_dynamic_policy_ids_config,9096400,5270307,LIU2021102480,stochastic_game_approach_2,Wang2021GameTheoreticAI,zhang2019,fog_computing_irs}. These games are formulated using various models from the game-theoretic literature. For example: Stochastic Games (\sgs) (see e.g. \cite{serkan_gyorgy_game,general_sum_markov_games_for_strategic_detection_of_apt,stocahstic_games_security_indep_nodes_nguyen_alpcan_basar,zhu_basar_dynamic_policy_ids_config,zhang2019,altman_jamming_1}), extensive-form games (see e.g. \cite{a_game_theoretic_ids_control_sys_alpcan_basar,nework_security_alpcan}), Blotto games (see e.g. \cite{9923774}), differential games (see e.g. \cite{9754705,fog_computing_irs}), hypergames (see e.g. \cite{8750848,9559403}), \posgs (see e.g. \cite{posg_cyber_deception_network_epidemic,hammar_stadler_cnsm_20,muzero_sdn,LIU2021102480,stochastic_game_approach_2}), Stackelberg games (see e.g. \cite{posg_cyber_deception_network_epidemic,optimal_thresholds_for_ids,5270307}), graph-based games (see e.g. \cite{game_resource_alloc_malicious_packet,Nguyen2018}), evolutionary games (see e.g. \cite{9754705,9096400}), continuous-kernel games (see e.g. \cite{a_game_theoretic_ids_control_sys_alpcan_basar}), rivalry games (see e.g. \cite{apt_rl_simulation}), and Bayesian games (see e.g. \cite{sengupta_marl_stackel_bayesian}).

This paper differs from the works referenced above in two main ways. First, we model the intrusion response use case as an optimal stopping game. The benefit of our model is that it provides insight into the structure of best response strategies through the theory of optimal stopping. Second, we evaluate obtained strategies on an emulated IT infrastructure. This contrasts with most of the prior works that use game-theoretic approaches, which either evaluate strategies analytically or in simulation \cite{flipit,dynamic_game_linan_zhu,general_sum_markov_games_for_strategic_detection_of_apt,honeypot_game,DBLP:journals/compsec/HorakBTKK19,game_theoretic_modeling_ofhoneypot_selection,serkan_gyorgy_game,posg_cyber_deception_network_epidemic,stocahstic_games_security_indep_nodes_nguyen_alpcan_basar, hammar_stadler_cnsm_20,optimal_thresholds_for_ids,a_game_theoretic_ids_control_sys_alpcan_basar,zhu_basar_dynamic_policy_ids_config,9096400,9754705,9559403,8750848,Nguyen2018,HUANG2020101660,Wang2021GameTheoreticAI,apt_rl_simulation,HAN2022113912,zhang2019,altman_jamming_1}.

Game-theoretic formulations based on optimal stopping theory can be found in prior research on Dynkin games \cite{dynkin_orig_3,dynkin_orig_2,dynkin_example_1,dynkin_example_2,dynkin_example_3}. Compared to these articles, our approach is more general by (\textit{i}) allowing each player to take multiple stop actions within an episode; and (\textit{ii}) by not assuming a game of perfect information. Another difference is that the referenced articles either study purely mathematical problems or problems in mathematical finance. To the best of our knowledge, we are the first to apply the stopping game formulation to the use case of intrusion response.

Our stopping game has similarities with the \flipit game \cite{flipit} and signaling games \cite{signaling_game_original}, both of which are commonplace in the security literature (see survey \cite{game_t_sec_survey} and textbooks \cite{nework_security_alpcan,tambe,carol_book_intrusion_detection,levente_book}). Signaling games have the same information asymmetry as our game and \flipit uses the same binary state space to model the state of an attack. The main differences are as follows. \flipit models the use case of advanced persistent threats and is a symmetric non-zero-sum game. In contrast, our game models an intrusion response use case and is an asymmetric zero-sum game. Lastly, compared to signaling games, the main difference is that our game is a sequential and simultaneous-move game. Signaling games, in comparison, are typically two-stage games where one player moves in each stage.

Previous game-theoretic studies that use emulation systems similar to ours are \cite{9328143} and \cite{5270307}. Specifically, in \cite{9328143}, a denial-of-service use case is formulated as a signaling game, for which a Nash equilibrium is derived. The equilibrium is then used to design a defense mechanism that is evaluated in a software-defined network emulation based on \mininet \cite{mininet}. Compared to this paper, the main differences are that we focus on a different use case than \cite{9328143} and that our solution method is based on reinforcement learning.

Similar to this paper, the authors of \cite{5270307} formulate an intrusion response use case as a \posg where the defender observes alerts from a Snort \idps \cite{snort}. In contrast to our approach, however, the approach of \cite{5270307} assumes access to attack-defense trees designed by human experts. Another difference between this paper and \cite{5270307} is the \posg. The \posg in \cite{5270307} has a larger state space than the \posg considered in this paper. Although this makes the \posg in \cite{5270307} more expressive than ours, it also makes computation of optimal defender strategies intractable. In fact, to estimate optimal defender strategies, the authors of \cite{5270307} are forced to approximate their model with one that has a smaller state space and is fully observed. In comparison, we are able to efficiently approximate equilibria of our game, without relying on model simplifications and without assuming access to attack-defense trees designed by human experts.
\subsection{Control Theory for Automated Intrusion Response}
Control theory provides a well-established mathematical framework for studying automatic systems. Classical control systems involve actuators in the physical world (e.g. electric power systems \cite{scada_control_example}) and many studies have focused on applying control theory to automate intrusion responses in cyber-physical systems (see surveys \cite{7011201,7011179,8795652}).

The control framework can also be applied to computing systems and interest in control theory among researchers in IT security is growing (see survey \cite{Miehling_control_theoretic_approaches_summary}). As opposed to classical control theory, which is focused on \textit{continuous-time} systems, the research on applying control theory to computing systems is focused almost entirely on \textit{discrete-time} systems. The main reason being that measurements from computer systems are solicited on a sampled basis, which is best described by a discrete-time model \cite{feedback_control_computing_systems,control_os_1}.

Previous works that apply control theory to the use case of intrusion response include: \cite{7568529,Kreidl2004FeedbackCA,7573127,miehling_attack_graph,miehling_control_security_2,miehling_control_security_3,miehling_control_security_4}. All of which model the problem of selecting response actions as the problem of controlling a discrete-time dynamical system and obtain optimal defender strategies through dynamic programming.

The main limitation of the works referenced above is that dynamic programming does not scale to problems of practical size due to the curse of dimensionality \cite{bellman_eq,BertsekasTsitsiklis96}.

\subsection{Reinforcement Learning for Automated Intrusion Response}
Reinforcement learning has emerged as a promising approach to approximate optimal control strategies in scenarios where exact dynamic programming is not applicable, and fundamental breakthroughs demonstrated by systems like \alphago in 2016 \cite{deepmind_2} and \openai \five in 2019 \cite{dota_openai_1} have inspired us and other researchers to study reinforcement learning with the goal to automate security functions (see surveys \cite{deep_rl_cyber_sec,control_rl_reviews}).

A large number of studies have focused on applying reinforcement learning to use cases similar to the intrusion response use case we discuss in this paper \cite{hammar_stadler_cnsm_20,hammar_stadler_cnsm_21, hammar_stadler_cnsm_22, elderman, schwartz_2020, oslo_pentest_rl, kurt_rl, microsoft_red_teaming, ridley_ml_defense, rl_cyberdefense_heartbleed, deep_hierarchical_rl_pentest, pentest_rl_rohit, adaptive_cyber_defense_pomdp_rl, muzero_sdn,atmos,sdn_rl_ddos,deep_air,noms_demo_preprint,game_cyber_rl_sim,al_shaer_book_ppo_simulation,cmu_ppo_selfplay,YUNGAICELANAULA2022103444,sdn_zolo,9923774,9916301, 9893186,rl_qcd_kalle,MAEDA2021102108,BLAND2020101738, HUANG2022102844,LIU2021102480,9096506,Khoury2020AHG,Wang2020AnID, han_yi_sdn,stochastic_game_approach_2,Iannucci2021AnIR,Iannucci2019APE,Wang2021GameTheoreticAI,hammar_stadler_tnsm,3560830.3563732,zhang2019}. These works use a variety of models, including \mdps \cite{oslo_pentest_rl,ridley_ml_defense,deep_hierarchical_rl_pentest,pentest_rl_rohit,deep_air,al_shaer_book_ppo_simulation,YUNGAICELANAULA2022103444,MAEDA2021102108,Iannucci2021AnIR,Iannucci2019APE,3560830.3563732}, \sgs \cite{elderman, hammar_stadler_cnsm_20, muzero_sdn,game_cyber_rl_sim,LIU2021102480,Wang2021GameTheoreticAI,model_free_rl_irs}, attack graphs \cite{cmu_ppo_selfplay}, Petri nets \cite{BLAND2020101738}, and \pomdps \cite{hammar_stadler_cnsm_21,adaptive_cyber_defense_pomdp_rl,kurt_rl,hammar_stadler_tnsm}, as well as various reinforcement learning algorithms, including Q-learning \cite{elderman,oslo_pentest_rl,ridley_ml_defense,mec_game_rl_q_learning_sim,9893186,BLAND2020101738,Wang2020AnID,3560830.3563732}, \sarsa \cite{kurt_rl}, \ppo \cite{hammar_stadler_cnsm_20,hammar_stadler_cnsm_21,cmu_ppo_selfplay,al_shaer_book_ppo_simulation,sdn_zolo}, hierarchical reinforcement learning \cite{deep_hierarchical_rl_pentest}, \dqn \cite{pentest_rl_rohit,YUNGAICELANAULA2022103444,sdn_zolo,9923774,9916301,LIU2021102480,Iannucci2021AnIR}, Thompson sampling \cite{adaptive_cyber_defense_pomdp_rl}, \muzero \cite{muzero_sdn}, \nfq \cite{atmos}, \ddqn \cite{deep_air,stochastic_game_approach_2}, \nfsp \cite{nfsp_jamming_1_sim,nfsp_security_2}, \ac \cite{MAEDA2021102108}, \acc \cite{han_yi_sdn}, and \ddpg \cite{sdn_rl_ddos,game_cyber_rl_sim}.

This paper differs from the works referenced above in three main ways. First, we model the intrusion response use case as a partially observed stochastic game. Most of the other works model the use case as a single-agent \mdp or \pomdp. The advantage of the game-theoretic model is that it allows finding defender strategies that are effective against a dynamic attacker, i.e. an attacker that adapts its strategy in response to the defender strategy.

Second, in a novel approach, we derive structural properties of optimal defender strategies in the game using optimal stopping theory.

Third, our method to find effective defender strategies includes using an emulation system in addition to a simulation system. The advantage of our method compared to the simulation-only approaches \cite{hammar_stadler_cnsm_20,hammar_stadler_cnsm_21, elderman, schwartz_2020, oslo_pentest_rl, kurt_rl, microsoft_red_teaming, ridley_ml_defense, rl_cyberdefense_heartbleed, deep_hierarchical_rl_pentest, pentest_rl_rohit, adaptive_cyber_defense_pomdp_rl,game_cyber_rl_sim,cmu_ppo_selfplay,al_shaer_book_ppo_simulation,mec_game_rl_q_learning_sim,nfsp_jamming_1_sim,9923774,9916301,HUANG2022102844,LIU2021102480,9096506,stochastic_game_approach_2,Iannucci2019APE,Iannucci2019APE,Wang2021GameTheoreticAI,3560830.3563732,zhang2019} is that the parameters of our simulation system are determined by measurements from an emulation system instead of being chosen by a human expert. Further, the learned strategies are evaluated in the emulation system, not in the simulation system. As a consequence, the evaluation results give higher confidence of the obtained strategies' performance in the target infrastructure than what simulation results would provide.

Some prior work on automated learning of security strategies that make use of emulation are: \cite{Wang2020AnID}, \cite{muzero_sdn}, \cite{atmos}, \cite{sdn_rl_ddos}, \cite{deep_air}, \cite{YUNGAICELANAULA2022103444}, \cite{Khoury2020AHG}, \cite{han_yi_sdn}, \cite{Iannucci2021AnIR}, \cite{beyond_cage}, and \cite{sdn_zolo}. They either emulate software-defined networks based on \mininet \cite{mininet} or use custom testbeds. The main differences between these efforts and the work described in this article are: (\textit{i}) we develop our own emulation system which allows for experiments with a large variety of exploits; (\textit{ii}) we focus on a different use case (most of the referenced works study denial-of-service attacks); (\textit{iii}) we do not assume that the defender has perfect observability; (\textit{iv}) we do not assume a static attacker; and (\textit{v}) we use an underlying theoretical framework to formalize the use case, derive structural properties of optimal strategies, and test these properties in an emulation system.

Finally, \cite{cyborg}, \cite{cygil}, and \cite{farland} describe efforts in building emulation platforms for reinforcement learning and cyber defense, which resemble our emulation system. In contrast to these articles, our emulation system has been built to investigate the specific use case of intrusion response and forms an integral part of our general solution method (see Fig. \ref{fig:method}).

\section{Conclusion and Future Work}\label{sec:conclusions}
In this work, we combine a formal framework with a practical evaluation to address the problem of automated intrusion response. We formulate the interaction between an attacker and a defender as an optimal stopping game. This formulation gives us insight into the structure of optimal strategies, which we prove to have threshold properties. Based on this knowledge, we develop a fictitious self-play algorithm, Threshold Fictitious Self-Play (\tfp), which learns near-optimal strategies in an efficient way. The results from running \tfp show that the learned strategies converge to an approximate Nash equilibrium and thus to near-optimal strategies (see Fig. \ref{fig:exploitability_curve}). The results also demonstrate that \tfp converges faster than a state-of-the-art fictitious self-play algorithm by taking advantage of threshold properties of optimal strategies (see Fig. \ref{fig:converge_times}). The threshold properties further enable us to provide a graphic representation of the learned strategies in a simple form (see Fig. \ref{fig:threshold_policy_3} and Fig. \ref{fig:stop_prob2}).

To assess the learned strategies in a real environment, we evaluate them in a system that emulates our target infrastructure (see Fig. \ref{fig:system2}). The results show that the strategies achieve almost the same performance in the emulated infrastructure as in the simulation. This gives us confidence that the obtained strategies would perform as expected in the target infrastructure, which is not feasible to evaluate directly.

We plan to continue this work in several directions. First, we will extend the current model of the attacker and the defender, which currently captures only timing of actions, to include decisions about a range of attacker and defender actions. Second, we plan to combine the strategies learned through our framework with techniques for online play, such as rollout \cite{bertsekas2021rollout}. Third, we plan to study techniques that allow to obtain defender strategies that generalize to a variety of infrastructure configurations and topologies. Fourth, we intend to extend our framework to an online-learning setting where the defender strategies co-evolve with changes in the target infrastructure.
\section{Acknowledgments}
This research has been supported in part by the Swedish armed forces and was conducted at KTH Center for Cyber Defense and Information Security (CDIS). The authors would like to thank Pontus Johnson for his useful input to this research, and Forough Shahab Samani and Xiaoxuan Wang for their constructive comments on a draft of this paper. The authors are also grateful to Branislav Bosansk{\'{y}} for sharing the code of the \hsvi algorithm for one-sided \posgs and to Jakob Stymne for contributing to our implementation of \nfsp.
\appendices
\section{Proofs}\label{appendix:proofs}
\subsection{Proof of Theorem \ref{thm:best_responses}.A}\label{appendix:theorem_1_a}
\begin{proof}[Proof.]
Since the \posg $\Gamma$ in (\ref{eq:game_def}) is finite and $\gamma \in (0,1)$, the existence proofs in \cite[\S 3]{posg_equilibria_existence_finite_horizon} and \cite[Thm. 2.3]{horak_thesis} apply, which state that a mixed Nash equilibrium exists. For the sake of brevity we do not restate the proofs, which are based on formulating the \posg as a finite strategic form game and appealing to Nash's theorem \cite[Thm. 1]{nash51}.

We prove that a pure Nash equilibrium exists when $s=0 \iff b(1) = 0$ using a proof by construction. It follows from (\ref{eq:reward_0})--(\ref{eq:reward_5}) and (\ref{eq:br_defender}) that the pure strategy defined by $\bar{\pi}_{\mathrm{D}}(0)=\mathfrak{C}$ and $\bar{\pi}_{\mathrm{D}}(b(1))=\mathfrak{S} \iff b(1)>0$ is a best response for the defender against any attacker strategy when $s=0 \iff b(1) = 0$. Similarly, given $\bar{\pi}_{\mathrm{D}}$, we conclude from (\ref{eq:reward_0})--(\ref{eq:reward_5}) and (\ref{eq:br_attacker}) that the pure strategy defined by $\bar{\pi}_{\mathrm{A}}(b(1), 0)=\mathfrak{C}$ and $\bar{\pi}_{\mathrm{A}}(b(1),1)=\mathfrak{S}$ for all $b(1) \in [0,1]$ is a best response for the attacker. Hence, $(\bar{\pi}_{\mathrm{D}},\bar{\pi}_{\mathrm{A}})$ is a pure Nash equilibrium (\ref{eq:minmax_objective}).
\end{proof}
\subsection{Proof of Theorem \ref{thm:best_responses}.B.}\label{appendix:theorem_1_b}
\begin{proof}[Proof.]
Given the \posg $\Gamma$ (\ref{eq:game_def}) and a fixed attacker strategy $\pi_{\mathrm{A}}$, any best response strategy for the defender $\tilde{\pi}_{\mathrm{D}}\in \mathscr{B}_{\mathrm{D}}(\pi_{\mathrm{A}})$ is an optimal strategy in a \pomdp $\mathcal{M}^{P}$ (see \S \ref{sec:game_analysis}). Hence, it is sufficient to show that there exists an optimal strategy $\pi_{\mathrm{D}}^{*}$ in $\mathcal{M}^{P}$ that satisfies (\ref{eq:prop_br_defender}). Conditions for (\ref{eq:prop_br_defender}) to hold and the existence proof are given in our previous work \cite{hammar_stadler_tnsm}[Thm 1.C]. Since $f_{O\mid s}$ is \tpp by assumption and all of the remaining conditions hold by definition of $\Gamma$ (\ref{eq:game_def}), the result follows.
\end{proof}
\subsection{Proof of Theorem \ref{thm:best_responses}.C.}\label{appendix:theorem_1_c}
Given the \posg $\Gamma$ (\ref{eq:game_def}) and a fixed defender strategy $\pi_{\mathrm{D}}$, any best response strategy for the attacker $\tilde{\pi}_{\mathrm{A}} \in \mathscr{B}_{\mathrm{A}}(\pi_{\mathrm{D}})$ is an optimal strategy in an \mdp $\mathcal{M}$ (see \S \ref{sec:game_analysis}). Hence, it is sufficient to show that there exists an optimal strategy $\pi_{\mathrm{A}}^{*}$ in $\mathcal{M}$ that satisfies (\ref{eq:prop_br_attacker_1})--(\ref{eq:prop_br_attacker_2}). To prove this, we use properties of $\mathcal{M}$'s value function $V_{\pi_{\mathrm{D}},l}^{*}$ (\ref{eq:bellman_eq_43}).

We use the value iteration algorithm to establish properties of $V_{\pi_{\mathrm{D}},l}^{*}$ \cite{puterman,krishnamurthy_2016}. Let $V_{\pi_{\mathrm{D}},l}^{k}$, $\mathscr{S}^{k,(\mathrm{A})}_{s,l,\pi_{\mathrm{D}}}$, and $\mathscr{C}^{k,(\mathrm{A})}_{s,l,\pi_{\mathrm{D}}}$, denote the value function, the stopping set (\ref{eq:stopping_set_at}), and the continuation set (\ref{eq:continue_set_at}) at iteration $k$ of the value iteration algorithm, respectively. Then, $\lim_{k\rightarrow \infty}V_{\pi_{\mathrm{D}},l}^k=V_{\pi_{\mathrm{D}},l}^{*}$, $\lim_{k\rightarrow \infty}\mathscr{S}^{k,(\mathrm{A})}_{s,l,\pi_{\mathrm{D}}}=\mathscr{S}^{(\mathrm{A})}_{s,l,\pi_{\mathrm{D}}}$, and $\lim_{k\rightarrow \infty}\mathscr{C}^{k,(\mathrm{A})}_{s,l,\pi_{\mathrm{D}}}$ $=\mathscr{C}^{(\mathrm{A})}_{s,l,\pi_{\mathrm{D}}}$ \cite[Thm 6.3.1]{puterman}\cite[Thm. 7.6.2-7.6.3]{krishnamurthy_2016}. We define $V_{\pi_{\mathrm{D}},l}^0\big($$(s,b(1))$$\big)=0$ for all $b(1)\in [0,1]$, $s \in \mathcal{S}$ and $l\in \{1,\hdots,L\}$.

Towards the proof of Theorem \ref{thm:best_responses}.C, we state the following six lemmas.
\begin{lemma}\label{lemma:V_geq_zero}
Given any defender strategy $\pi_{\mathrm{D}}$, $V^{*}_{\pi_{\mathrm{D}},l}\big($$s,b(1)\big)$ $\geq 0$ for all $s \in \mathcal{S}$ and $b(1) \in [0,1]$.
\end{lemma}
\begin{proof}[Proof.]
Consider $\bar{\pi}_{\mathrm{A}}$ defined by $\bar{\pi}_{\mathrm{A}}(0,\cdot)=\mathfrak{C}$ and $\bar{\pi}_{\mathrm{A}}(1,\cdot)=\mathfrak{S}$. Then it follows from (\ref{eq:reward_0})--(\ref{eq:reward_5}) that for any $\pi_{\mathrm{D}} \in \Pi_{\mathrm{D}}$, any $s \in \mathcal{S}$, and any $b(1) \in [0,1]$, the following holds: $V^{\bar{\pi}_{\mathrm{A}}}_{\pi_{\mathrm{D}},l}(s,b(1)) \geq 0$. By optimality, $V^{\bar{\pi}_{\mathrm{A}}}_{\pi_{\mathrm{D}},l}(s,b(1)) \leq V^{*}_{\pi_{\mathrm{D}},l}(s,b(1))$. Hence, $V^{*}_{\pi_{\mathrm{D}},l}(s,b(1)) \geq 0$.
\end{proof}

\begin{lemma}\label{lemma:decreasing_v}
$V^{*}_{\pi_{\mathrm{D}},l}\big($$b(1),1\big)$ is non-increasing with $\pi_{\mathrm{D}}(\mathfrak{S} \mid b(1))$ and non-decreasing with $l \in \{1,\hdots,L\}$.
\end{lemma}
\begin{proof}[Proof.]
We prove this statement by mathematical induction. For $k=1$, we know from (\ref{eq:reward_0})--(\ref{eq:reward_5}) that $V^1_{\pi_{\mathrm{D}},l}\big($$1$, $b(1)\big)$ is non-increasing with $\pi_{\mathrm{D}}(\mathfrak{S} \mid b(1))$ and non-decreasing with $l$.

For $k > 1$, $V^{k}_{\pi_{\mathrm{D}},l}$ is given by:
\begin{align}
  &V^{k}_{\pi_{\mathrm{D}},l}\big(b(1),1\big) = \max\Big[0, -\mathcal{R}\big(1,(\mathfrak{C},a^{(\mathrm{D})})\big) \label{eq:decreasing_v}\\
  &+ (1-\phi_{l})\sum_of_O(o\mid 1)V^{k-1}_{l-a^{(\mathrm{D})}}\big(b(1),1\big)\Big]\nonumber
\end{align}
The first term inside the maximization in (\ref{eq:decreasing_v}) is trivially non-increasing with $\pi_{\mathrm{D}}(\mathfrak{S} \mid b(1))$ and non-decreasing with $l$. Assume by induction that the statement of Lemma \ref{lemma:decreasing_v} holds for $V^{k-1}_{\pi_{\mathrm{D}},l}\big($$s$, $b(1)\big)$. Then the second term inside the maximization in (\ref{eq:decreasing_v}) is non-increasing with $\pi_{\mathrm{D}}(\mathfrak{S}\mid b(1))$ and non-decreasing with $l$ by (\ref{eq:reward_0})--(\ref{eq:reward_5}) and the induction hypothesis. Hence, $V^{k}_{\pi_{\mathrm{D}},l}\big($$s$, $b(1)\big)$ is non-increasing with $\pi_{\mathrm{D}}(\mathfrak{S}\mid b(1))$ and non-decreasing with $l$ for all $k\geq 0$.
\end{proof}

\begin{lemma}\label{lemma:dec_stopping}
If $f_{O}$ is \tpp and $\pi_{\mathrm{D}}(\mathfrak{S} \mid b(1))$ is increasing with $b(1)$, then $V_{\pi_{\mathrm{D}},l}(b(1),1) \geq \sum_{o}f_O(o \mid 1)V_{\pi_{\mathrm{D}},l}(b^{o}(1),1)$, where $b^{o}(1)$ denotes $b(1)$ updated with (\ref{eq:belief_upd}) after observing $o \in \mathcal{O}$.
\end{lemma}
\begin{proof}[Proof.]
Since $f_O$ is \tpp, it follows from \cite[Thm. 10.3.1, pp. 225 and 238]{krishnamurthy_2016} and \cite[Lemma 4, pp. 12]{hammar_stadler_tnsm} that given two beliefs $b^{\prime}(1) \geq b(1)$ and given two observations $o \geq \bar{o}$, the following holds for any $k \in \mathcal{O}$ and $l \in \{1,\hdots, L\}$: $b^{\prime,o}(1) \geq b^{o}(1)$, $\mathbb{P}[o \geq k \mid b^{\prime}(1)] \geq \mathbb{P}[o \geq k \mid b(1)]$, and $b_{a}^{o}(1) \geq b_{a}^{\bar{o}}(1)$.

Since $\pi_{\mathrm{D}}$ is increasing with $b(1)$ and $V_{\pi_{\mathrm{D}},l}(b(1),1)$ is decreasing with $b(1)$ (Lemma \ref{lemma:decreasing_v}), it follows that $\mathbb{E}_{o}[b^{o}(1)] \geq b(1)$, and thus $V_{\pi_{\mathrm{D}},l}(b(1),1) \geq \sum_{o}f_O(o\mid 1)V_{\pi_{\mathrm{D}},l}(b^{o}(1),1)$.
\end{proof}

\begin{lemma}\label{lemma:stop_if_defender_knows}
If $f_{O}$ is \tpp, $\pi_{\mathrm{D}}(\mathfrak{S} \mid b(1))=1$, and $\pi_{\mathrm{D}}(\mathfrak{S} \mid b(1))$ is increasing with $b(1)$, then $V^{*}_{\pi_{\mathrm{D}},l}\big($$s,$$b(1)\big)$$=0$ and for any $\tilde{\pi}_{\mathrm{A}} \in \mathscr{B}_{\mathrm{A}}(\pi_{\mathrm{D}})$, $\tilde{\pi}_{\mathrm{A}}$$(1,$$b(1))$$=\mathfrak{S}$.
\end{lemma}
\begin{proof}[Proof.]
From (\ref{eq:bellman_eq_41})--(\ref{eq:bellman_posg_1}) we know that $\tilde{\pi}_{\mathrm{A}}(b(1),1)=\mathfrak{S}$ iff:
\begin{align}
\frac{\mathrm{R}_{\mathrm{st}}}{l} + (\phi_l-1)\sum_of_O(o\mid 1)V^{*}_{\pi_{\mathrm{D}},l-a^{(1)}}(b^{o}(1),1) \geq 0 \label{eq:ineq_stop_if_def}
\end{align}
We know that $\mathrm{R}_{\mathrm{st}} \geq 0$ (see \S \ref{sec:formal_model_2}). Further, since $f_{O}$ is \tpp, $\pi_{\mathrm{D}}(\mathfrak{S} \mid b(1))=1$, and since $\pi_{\mathrm{D}}(\mathfrak{S} \mid b(1))$ is increasing with $b(1)$, we have by Lemma \ref{lemma:dec_stopping} that $\mathbb{E}_{o}[\pi_{\mathrm{D}}(\mathfrak{S} \mid b^{o}(1))]=1$. As a consequence, the second term in the left-hand side of (\ref{eq:ineq_stop_if_def}) is zero. Hence, the inequality holds and $\tilde{\pi}_{\mathrm{A}}(b(1),1)=1$, which implies that $V^{*}_{\pi_{\mathrm{D}},l}\big($$s,b(1)\big)=0$.
\end{proof}

\begin{lemma}\label{lemma:continue_stop_equality}
Given any defender strategy $\pi_{\mathrm{D}}\in \Pi_{\mathrm{D}}$, if $\pi^{*}_{\mathrm{A}}(b(1),1) = \mathfrak{S}$, then $\pi^{*}_{\mathrm{A}}(b(1),0) = \mathfrak{C}$.
\end{lemma}
\begin{proof}[Proof.]
$\pi^{*}_{\mathrm{A}}$$(1,$$b(1))$$=\mathfrak{S}$ implies that $V^{*}_{\pi_{\mathrm{D}},l}$$(1,$$b(1))$$=0$. Hence, it follows from Lemma \ref{lemma:dec_stopping} that:
\begin{align}
 &(1-\phi_l)\sum_{o\in \mathcal{O}} f_O(o\mid 1)V^{*}_{\pi_{\mathrm{D}},l}(b^{o},1) \leq 0\\
&\implies \sum_{o\in \mathcal{O}} f_O(o \mid 1)V^{*}_{\pi_{\mathrm{D}},l}(b^{o},1) \leq \sum_{o\in \mathcal{O}} f_O(o \mid 0)V^{*}_{\pi_{\mathrm{D}},l}(b^{o},0)\nonumber\\
&\implies \pi^{*}_{\mathrm{A}}(b(1),0) = \mathfrak{C}\nonumber
\end{align}
\end{proof}

\begin{lemma}\label{lemma:decreasing_differences}
If $\pi_{\mathrm{D}}(\mathfrak{S}\mid b(1))$ is non-decreasing with $b(1)$ and $f_O$ is \tpp, then $V^{*}_{\pi_{\mathrm{D}},l}\big($$b(1),0\big) - V^{*}_{\pi_{\mathrm{D}},l}\big($$b(1),1\big)$ is non-decreasing with $\pi_{\mathrm{D}}(\mathfrak{S} \mid b(1))$.
\end{lemma}
\begin{proof}[Proof.]
We prove this statement by mathematical induction. Let $W^{k}_{\pi_{\mathrm{D}},l}(b(1))$ $= V^{k}_{\pi_{\mathrm{D}},l}\big($$b(1),0\big) - V^{k}_{\pi_{\mathrm{D}},l}\big($$b(1),1\big)$. For $k=1$, it follows from (\ref{eq:reward_0})--(\ref{eq:reward_5}) that $W^{1}_{\pi_{\mathrm{D}},l}(b(1))$ is non-decreasing with $\pi_{\mathrm{D}}(\mathfrak{S} \mid b(1))\in [0,1]$. Assume by induction that the statement of Lemma \ref{lemma:decreasing_differences} holds for $W^{k-1}_{\pi_{\mathrm{D}},l}(b(1))$. We show that then the statement holds also for $W^{k}_{\pi_{\mathrm{D}},l}(b(1))$.

There are three cases to consider:
\begin{itemize}
\item If $b(1) \in \mathscr{S}^{k,(\mathrm{A})}_{0,l,\pi_{\mathrm{D}}} \cap \mathscr{C}^{k,(\mathrm{A})}_{1,l,\pi_{\mathrm{D}}}$, then:
\begin{align}
  &W^{k}_{\pi_{\mathrm{D}},l}(b(1)) = \mathrm{R}_{\mathrm{int}} + \label{eq:ind_case_1}\\
  &\pi_{\mathrm{D}}(\mathfrak{S} \mid b(1))\left(\frac{\mathrm{R}_{\mathrm{st}}}{l} -\frac{\mathrm{R}_{\mathrm{cost}}}{l} - \mathrm{R}_{\mathrm{int}}\right)\nonumber
\end{align}
The right-hand side of (\ref{eq:ind_case_1}) is non-decreasing with $\pi_{\mathrm{D}}(\mathfrak{S} \mid b(1))$ since $\frac{\mathrm{R}_{\mathrm{st}}}{l}- \frac{\mathrm{R}_{\mathrm{cost}}}{l}- \mathrm{R}_{\mathrm{int}}\geq 0$ (see \S \ref{sec:formal_model_2}).
\item If $b(1) \in \mathscr{C}^{k,(\mathrm{A})}_{0,l,\pi_{\mathrm{D}}} \cap \mathscr{C}^{k,(\mathrm{A})}_{1,l,\pi_{\mathrm{D}}}$, then using (\ref{eq:tp_2}) and (\ref{eq:belief_upd}):
\begin{align}
  &W^{k}_{\pi_{\mathrm{D}},l}(b(1)) = \pi_{\mathrm{D}}(\mathfrak{S} \mid b(1))\Big(\frac{\mathrm{R}_{\mathrm{st}}}{l} - \frac{\mathrm{R}_{\mathrm{cost}}}{l} \label{eq:ind_case_2}\\
  &-\mathrm{R}_{\mathrm{int}}\Big) + V^{k-1}_{\pi_{\mathrm{D}},l}\big(b(1),1\big)\Big) + \mathrm{R}_{\mathrm{int}} +\sum_o\Big(f_{O}(o\mid 0)\cdot\nonumber\\
  &V^{k}_{\pi_{\mathrm{D}},l}\big(b^{o}(1),0\big) - (1-\phi_l)f_{O}(o\mid 1)V^{k}_{\pi_{\mathrm{D}},l}\big(b^{o}(1),1\big)\Big)\nonumber
\end{align}
The first term in the right-hand side of (\ref{eq:ind_case_2}) is non-decreasing with $\pi_{\mathrm{D}}(\mathfrak{S} \mid b(1))$ since $\frac{\mathrm{R}_{\mathrm{st}}}{l} - \frac{\mathrm{R}_{\mathrm{cost}}}{l}- \mathrm{R}_{\mathrm{int}}\geq 0$ (see \S \ref{sec:formal_model_2}) and $V^{k-1}_{\pi_{\mathrm{D}},l}\big(b(1),1\big) \geq 0$ (it is a consequence of Lemma \ref{lemma:V_geq_zero} and (\ref{eq:bellman_eq_41})--(\ref{eq:bellman_posg_1})). The second term is non-decreasing with $\pi_{\mathrm{D}}(\mathfrak{S} \mid b(1))$ by the induction hypothesis and the assumption that $f_O$ is \tpp.
\item If $b(1) \in \mathscr{C}^{k,(\mathrm{A})}_{0,l,\pi_{\mathrm{D}}} \cap \mathscr{S}^{k,(\mathrm{A})}_{1,l,\pi_{\mathrm{D}}}$, then:
\begin{align}
  &W^{k}_{\pi_{\mathrm{D}},l}(b(1)) = \pi_{\mathrm{D}}(\mathfrak{S} \mid b(1))\left(-\frac{\mathrm{R}_{\mathrm{cost}}}{l}\right) \label{eq:step_1_1}\\
  & + \sum_of_{O}(o \mid 0)V^{k}_{\pi_{\mathrm{D}},l}\big(b^{o}(1),0\big) \nonumber\\
  &=\pi_{\mathrm{D}}(\mathfrak{S} \mid b(1))\left(-\frac{\mathrm{R}_{\mathrm{cost}}}{l}\right) + \sum_o\Big(f_{O}(o \mid 0)\cdot \label{eq:step_1_2}\\
  &V^{k}_{\pi_{\mathrm{D}},l}\big(b^{o}(1),0\big) - (1-\phi_l)f_{O}(o \mid 1)V^{k}_{\pi_{\mathrm{D}},l}\big(b^{o}(1),1\big)\Big)\nonumber
\end{align}
The first term in the right-hand side of (\ref{eq:step_1_1}) is non-decreasing with $\pi_{\mathrm{D}}(\mathfrak{S} \mid b(1))$ since $-\frac{\mathrm{R}_{\mathrm{cost}}}{l}\geq 0$. The second term is non-decreasing with $\pi_{\mathrm{D}}(\mathfrak{S} \mid b(1))$ by the induction hypothesis and the assumption that $f_O$ is \tpp. (\ref{eq:step_1_2}) follows from Lemma \ref{lemma:dec_stopping} and the fact that $b(1) \in \mathscr{S}^{k,(\mathrm{A})}_{1,l,\pi_{\mathrm{D}}}$.
\end{itemize}
The other cases, e.g. $b(1) \in \mathscr{S}^{k,(\mathrm{A})}_{0,l,\pi_{\mathrm{D}}} \cap \mathscr{S}^{k,(\mathrm{A})}_{1,l,\pi_{\mathrm{D}}}$, can be discarded due to Lemma \ref{lemma:continue_stop_equality}. Hence, $W^{k}_{\pi_{\mathrm{D}},l}(b(1))$ is non-decreasing with $\pi_{\mathrm{D}}(\mathfrak{S} \mid b(1))$ for all $k \geq 0$.
\end{proof}
We now use Lemmas \ref{lemma:V_geq_zero}-\ref{lemma:decreasing_differences} to prove Theorem \ref{thm:best_responses}.C. The main idea behind the proof is to show that the stopping sets in state $s=1$ have the form: $\mathscr{S}^{(\mathrm{A})}_{1,l,\pi_{\mathrm{D}}} = [\tilde{\beta}_{1,l},1]$, and that the continuation sets in state $s=0$ have the form: $\mathscr{C}^{(\mathrm{A})}_{0,l,\pi_{\mathrm{D}}} = [\tilde{\beta}_{0,l},1]$, for some values $\tilde{\beta}_{0,1}, \tilde{\beta}_{1,1}, \hdots, \tilde{\beta}_{0,L}, \tilde{\beta}_{1,L} \in [0,1]$.
\begin{proof}[Proof of Theorem \ref{thm:best_responses}.C.]
We first show that $1 \in \mathscr{S}^{(\mathrm{A})}_{1,l,\pi_{\mathrm{D}}}$ and that $1 \in \mathscr{C}^{(\mathrm{A})}_{0,l,\pi_{\mathrm{D}}}$. Since $\pi_{\mathrm{D}}(\mathfrak{S} \mid 1)=1$, it follows from Lemma \ref{lemma:stop_if_defender_knows} that $1 \in \mathscr{S}^{(\mathrm{A})}_{1,l,\pi_{\mathrm{D}}}$ and as a consequence of (\ref{eq:bellman_eq_41})--(\ref{eq:bellman_posg_1}) we have that $\tilde{\pi}_{\mathrm{A}}(b(1),0)=\mathfrak{C}$ iff:
\begin{align}
  &\sum_of_O(o\mid 0) V^{*}_{\pi_{\mathrm{D}},l-1}(b^{o}(1),0) -\nonumber\\
  &f_O(o \mid 1)V^{*}_{\pi_{\mathrm{D}},l-1}(b^{o}(1),1)\geq 0
\end{align}
The left-hand side of the above equation is positive since a) $f_{O}$ is assumed to be \tpp; b) $\sum_of_O(o\mid 0)V^{*}_{\pi_{\mathrm{D}},l-1}(b^{o}(1),0)\geq 0$ (Lemma \ref{lemma:V_geq_zero}); and c) $f_O(o \mid 1)V^{*}_{\pi_{\mathrm{D}},l-1}(1,$$b^{o}(1))=0$ (Lemma \ref{lemma:dec_stopping}). Hence, $1 \in \mathscr{C}^{(\mathrm{A})}_{0,l,\pi_{\mathrm{D}}}$.

Now we show that $\mathscr{S}^{(\mathrm{A})}_{1,l,\pi_{\mathrm{D}}} = [\tilde{\beta}_{1,l},1]$ and that $\mathscr{C}^{(\mathrm{A})}_{0,l,\pi_{\mathrm{D}}} = [\tilde{\beta}_{0,l},1]$ for some values $\tilde{\beta}_{0,1}, \tilde{\beta}_{1,1}, \hdots, \tilde{\beta}_{0,L}, \tilde{\beta}_{1,L} \in [0,1]$. From (\ref{eq:bellman_eq_41})--(\ref{eq:bellman_posg_1}) we know that $\tilde{\pi}_{2}(b(1),1)=\mathfrak{S}$ iff:
\begin{align}
&\mathbb{E}_{\pi_{\mathrm{D}}}\Big[\mathcal{R}\big(1,(a^{(\mathrm{D})}, \mathfrak{C})\big) + \label{eq:exp_1_1_1}\\
  &(\phi_{l}-1)\sum_of_O(o \mid 1)V^{*}_{\pi_{\mathrm{D}},l-a^{(\mathrm{D})}}(b^{o}(1),1)\Big] \geq 0\nonumber
\end{align}
The first term in the left-hand side of (\ref{eq:exp_1_1_1}) is increasing with $b(1)$ (\ref{eq:reward_0})--(\ref{eq:reward_5}). Further, it follows from Lemma \ref{lemma:decreasing_v} that the second term is decreasing with $b(1)$. Hence, we conclude that if $\tilde{\pi}_{\mathrm{A}}(b(1),1)=\mathfrak{S}$, then for any $b^{\prime}(1) \geq b(1)$, $\tilde{\pi}_{\mathrm{A}}(b^{\prime}(1),1)=\mathfrak{S}$. As a consequence, there exist values $\tilde{\beta}_{1,1}, \hdots, \tilde{\beta}_{1,L}$ such that $\mathscr{S}^{(\mathrm{A})}_{1,l,\pi_{\mathrm{D}}} = [\tilde{\beta}_{1,l},1]$.

Similarly, from (\ref{eq:bellman_eq_41})--(\ref{eq:bellman_posg_1}) we know that $\tilde{\pi}_{\mathrm{A}}(b(1),0)=\mathfrak{C}$ iff:
\begin{align}
&\mathbb{E}_{\pi_{\mathrm{D}}}\Big[\sum_of_O(o \mid 0)V^{*}_{\pi_{\mathrm{D}},l-a^{(\mathrm{D})}}(b^{o}(1),0)\label{eq:theorem_b_12}\\
&-f_O(o \mid 1)V^{*}_{\pi_{\mathrm{D}},l-a^{(\mathrm{D})}}(b^{o}(1),1)\Big] \geq 0\nonumber
\end{align}
Since $f_{O}$ is \tpp and $\pi_{\mathrm{D}}(\mathfrak{S} \mid b(1))$ is increasing with $b(1)$, the left-hand side in (\ref{eq:theorem_b_12}) is decreasing (it follows from Lemma \ref{lemma:decreasing_v} and Lemma \ref{lemma:decreasing_differences}). Hence, we conclude that if $\tilde{\pi}_{\mathrm{A}}(b(1),0)=\mathfrak{C}$, then for any $b^{\prime}(1) \geq b(1)$, $\tilde{\pi}_{\mathrm{A}}(b^{\prime}(1),0)=\mathfrak{C}$. As a result, there exist values $\tilde{\beta}_{0,1}, \hdots, \tilde{\beta}_{0,L}$ such that $\mathscr{C}^{(\mathrm{A})}_{0,l,\pi_{\mathrm{D}}} = [\tilde{\beta}_{0,l},1]$.
\end{proof}
\section{Hyperparameters}\label{appendix:hyperparameters}
\begin{table}
\centering
\resizebox{1\columnwidth}{!}{%
  \begin{tabular}{ll} \toprule
  \textbf{Game Parameters} & {\textbf{Values}} \\
  \hline
  $\mathrm{R}_{\mathrm{st}}, \mathrm{R}_{\mathrm{cost}}, \mathrm{R}_{\mathrm{int}}$,$\gamma$, $\phi_{l}$, $L$ & $20$, $-2$, $-1$, $0.99$, $1/2l$, $7$\\
  {\textbf{\tfp Parameters}} & {\textbf{Values}} \\
  \hline
  $c, \epsilon, \lambda, A, a, N, \delta$ & $10$, $0.101$, $0.602$, $100$, $1$, $50$, $0.2$\\
  {\textbf{\nfsp Parameters}} & {\textbf{Values}} \\
  \hline
  lr \rll, lr \sll, batch, \# layers & $10^{-2}$,$5\cdot 10^{-3}$, $64$, $2$ \\
  \# neurons, $\mathcal{M}_{RL}$, $\mathcal{M}_{SL}$ & $128$, $2\times 10^{5}$, $2\times 10^{6}$,\\
  $\epsilon$, $\epsilon$-decay, $\eta$ & $0.06$, $0.001$, $0.1$\\
  {\textbf{\hsvi Parameter}} & {\textbf{Value}} \\
  \hline
  $\epsilon$ & $3$\\
  \bottomrule\\
\end{tabular}
}
\caption{Hyperparameters of the \posg and the algorithms used for evaluation.}\label{tab:hyperparams}
\end{table}
The hyperparameters used for the evaluation in this paper are listed in Table \ref{tab:hyperparams} and were obtained through grid search.

\section{Configuration of the Infrastructure in Fig. \ref{fig:system2}}\label{appendix:infrastructure_configuration}
The configuration of the target infrastructure shown in Fig. \ref{fig:system2} is available in Table \ref{tab:emulation_setup}.
\begin{table}
\centering
\resizebox{1\columnwidth}{!}{%
\begin{tabular}{ll} \toprule
  {\textit{ID (s)}} & {\textit{OS:Services:Exploitable Vulnerabilities}} \\ \midrule
  $N_1$ & \ubuntu20:\snort(community ruleset v2.9.17.1),\ssh:- \\
  $N_2$ & \ubuntu20:\ssh,\http Erl-Pengine,\dns:\cwe-1391\\
  $N_4$ & \ubuntu20:\http \flask,\telnet,\ssh:\cwe-1391 \\
  $N_{10}$ &\ubuntu20:\ftp,\mongo,\smtp,\tomcat,\ts3,\ssh:\cwe-1391 \\
  $N_{12}$ & \jessie:\ts3,\tomcat,\ssh:\cve-2010-0426,\cwe-1391 \\
  $N_{17}$ & \wheezy:\apache2,\snmp,\ssh:\cve-2014-6271 \\
  $N_{18}$ & \debian 9.2:\irc,\apache2,\ssh:\cwe-89 \\
  $N_{22}$ & \jessie:\proftp,\ssh,\apache2,\snmp:\cve-2015-3306 \\
  $N_{23}$ & \jessie:\apache2,\smtp,\ssh:\cve-2016-10033 \\
  $N_{24}$ & \jessie:\ssh:\cve-2015-5602,\cwe-1391 \\
  $N_{25}$ & \jessie: \elastic,\apache2,\ssh,\snmp:\cve-2015-1427\\
  $N_{27}$ & \jessie:\samba,\ntp,\ssh:\cve-2017-7494\\
  $N_3$,$N_{11}$,$N_{5}$-$N_9$& \ubuntu20:\ssh,\snmp,\postgres,\ntp:-\\
  $N_{13-16}$,$N_{19-21}$,$N_{26}$,$N_{28-31}$& \ubuntu20:\ntp, \irc, \snmp, \ssh, \postgres:-\\
  \bottomrule\\
\end{tabular}
}
\caption{Configuration of the target infrastructure (Fig. \ref{fig:system2}).}\label{tab:emulation_setup}
\end{table}

\section{Distributions of Infrastructure Metrics}\label{appendix:infrastructure_metrics}
The emulation system (see Fig. \ref{fig:method}) collects hundreds of metrics every time-step. To measure the information that a metric provides for detecting intrusions, we calculate the Kullback-Leibler (\kl) divergence $D_{\mathrm{KL}}(f_{O \mid 0} \parallel f_{O \mid 1})$ between the distribution of the metric when no intrusion occurs $f_{O\mid s=0}$ and during an intrusion $f_{O\mid s=1}$:
\begin{align}
D_{\mathrm{KL}}(f_{O \mid 0} \parallel f_{O \mid 1}) &= \sum_{o \in \mathcal{O}}f_{O \mid 0}(o)\log\left(\frac{f_{O \mid 0}(o)}{f_{O \mid 1}(o)}\right)
\end{align}
Here $O$ denotes the random variable representing the value of the metric and $\mathcal{O}$ is the domain of $O$.

Figure \ref{fig:observations} shows empirical distributions of the collected metrics with the largest \kl divergence. We see that the \textsc{idps} alerts have the largest \kl divergence and thus provide the most information for detecting intrusions.
\begin{figure}
  \centering
    \scalebox{0.54}{
      \includegraphics{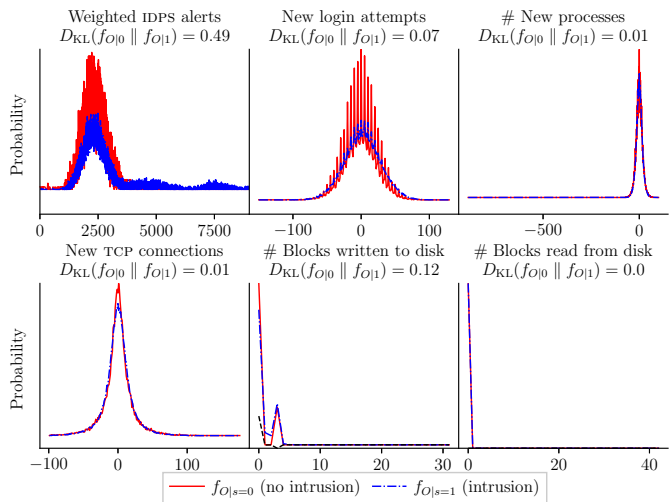}
    }
    \caption{Empirical distributions of selected infrastructure metrics; the red and blue lines show the distributions when no intrusion occurs and during intrusion, respectively.}
    \label{fig:observations}
\end{figure}

\section{Attacker Actions}\label{appendix:attacker actions}
The attacker actions and their descriptions are listed in Table \ref{tab:attacker_actions_descr}.
\begin{table}
\centering
\resizebox{1\columnwidth}{!}{%
\begin{tabular}{ll} \toprule
  {\textit{Action}} & {\textit{Description}} \\ \midrule
  \tcpp scan & \tcpp port scan by sending \syn or empty packets (\texttt{nmap}) \\
  \udp port scan & \udp port scan by sending \udp packets (\texttt{nmap}) \\
  ping scan & IP scan with \icmp ping messages \\
  \vulscan & vulnerability scan using \texttt{nmap}\\
  brute-force attack & performs a dictionary attack against a login service (\texttt{nmap})\\
  \cve-2017-7494 exploit & uploads malicious binary to the \samba service and executes it \\
  \cve-2015-3306 exploit & uses the \texttt{mod\_copy} in \texttt{proftpd} for remote code execution  \\
  \cve-2014-6271 exploit & uses a vulnerability in \texttt{bash} for remote code execution \\
  \cve-2016-10033 exploit & uses \texttt{phpmailer} for remote code execution \\
  \cve-2015-1427 exploit & uses \texttt{elasticsearch} for remote code execution \\
  \cwe-89 exploit & injects malicious SQL code to execute code remotely \\
  \bottomrule\\
\end{tabular}
}
\caption{Descriptions of the attacker actions.}\label{tab:attacker_actions_descr}
\end{table}

\ifCLASSOPTIONcaptionsoff
  \newpage
\fi

\bibliographystyle{IEEEtran}
\bibliography{references,url}

\begin{IEEEbiography}
    [{\includegraphics[width=1in,height=1.25in,clip,keepaspectratio]{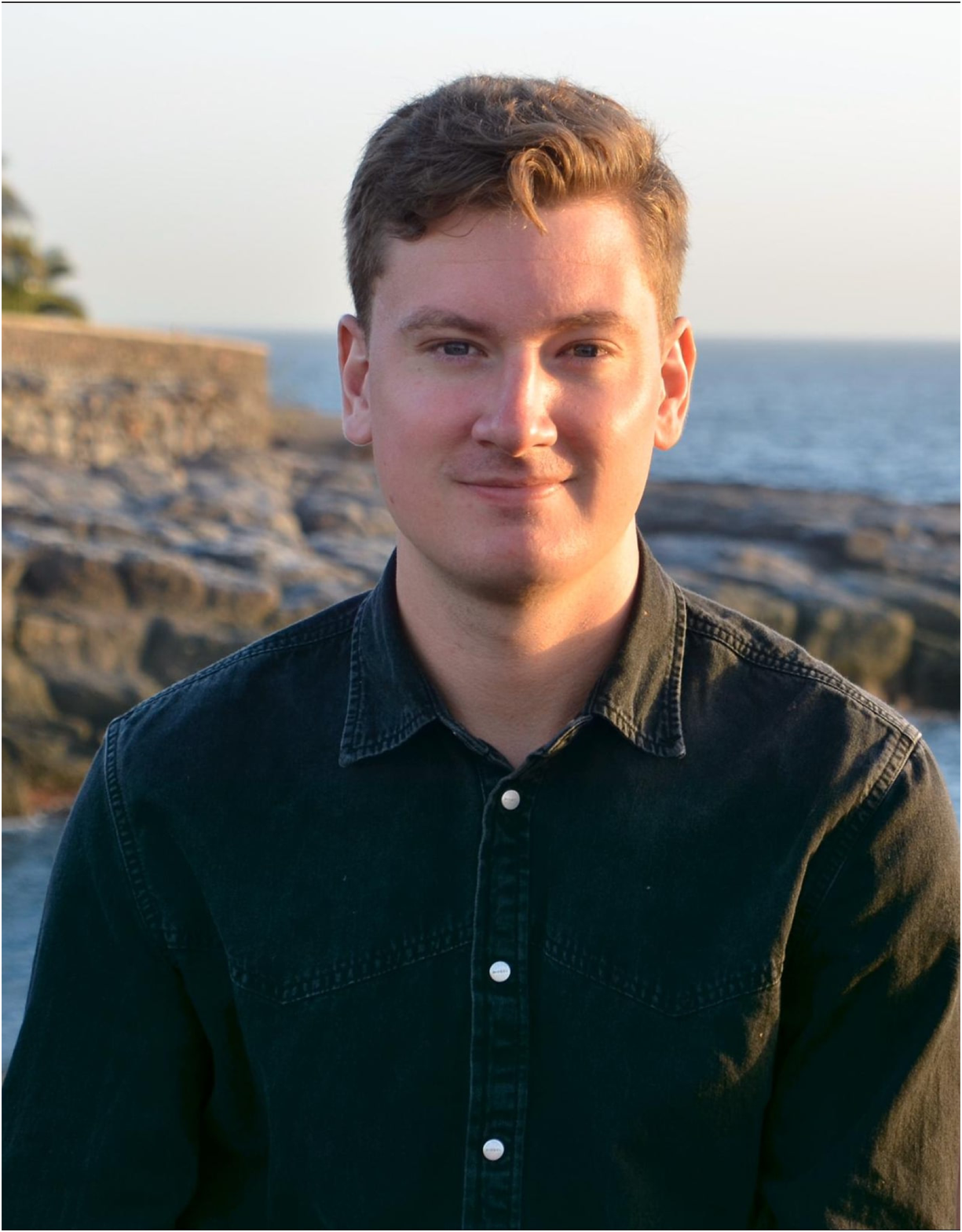}}]{Kim Hammar}
Kim Hammar is currently pursuing the Ph.D. degree at the Division of Network and Systems Engineering at KTH Royal Institute of Technology in Stockholm, Sweden. Before starting his Ph.D., he received the B.Sc. and M.Sc. degree in computer engineering with a specialization in distributed systems from KTH Royal Institute of Technology in 2016 and 2018, respectively. His research interests are in the intersection between decision theory, machine learning, and large-scale systems, focusing on cybersecurity applications.
\end{IEEEbiography}
\begin{IEEEbiography}
    [{\includegraphics[width=1in,height=1.25in,clip,keepaspectratio]{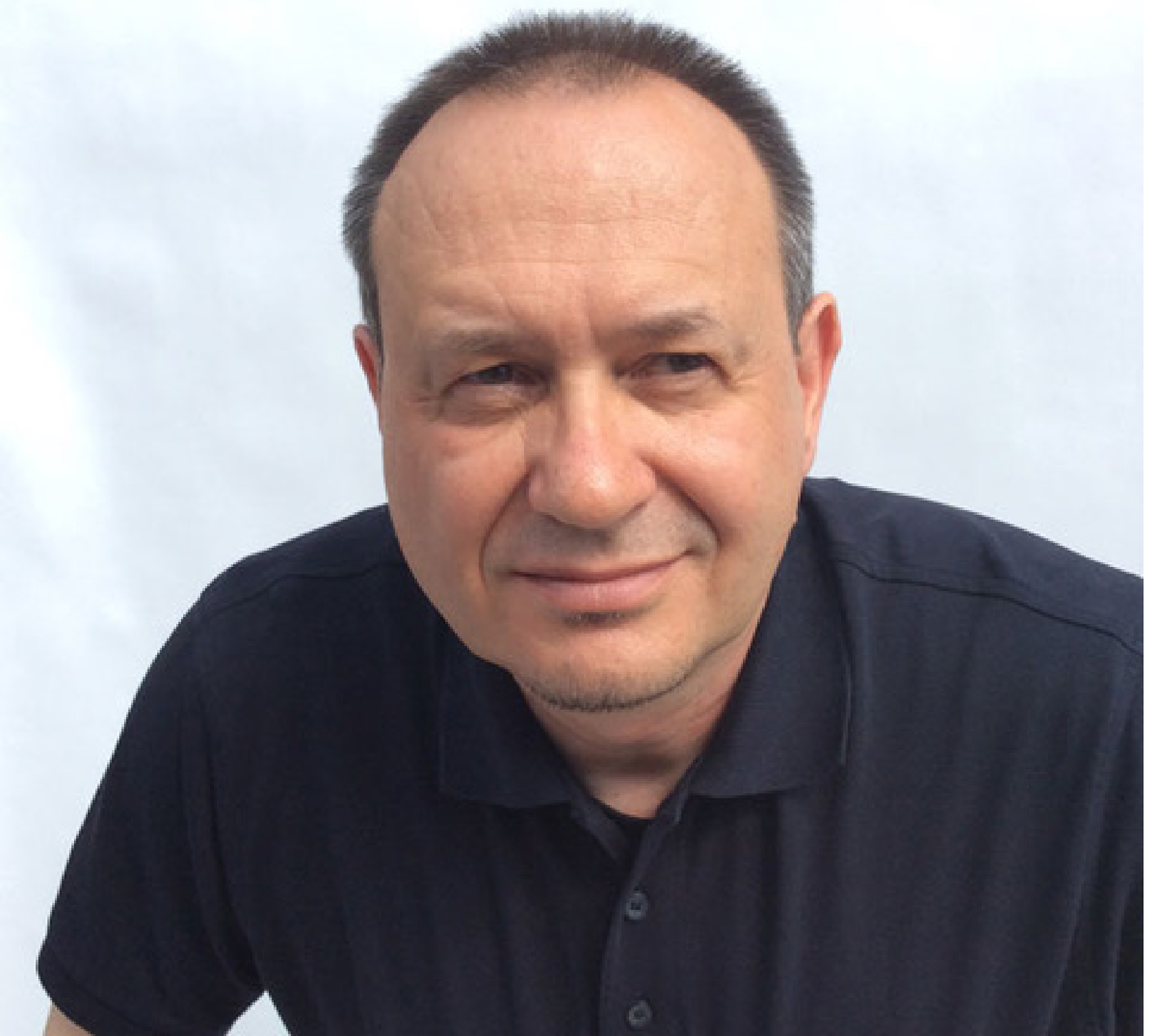}}]{Rolf Stadler}
Rolf Stadler is a professor at KTH Royal Institute of Technology in Stockholm, Sweden, and head of the Division of Network and Systems Engineering. He holds an M.Sc. degree in mathematics and a Ph.D. in computer science from the University of Zurich. Before joining KTH in 2001, he held positions at the IBM Zurich Research Laboratory, Columbia University, and ETH Zürich. His group made contributions to real-time monitoring, resource management, and automation for large-scale networked systems. His current interests include data-driven methods for network engineering and management, as well as AI techniques for cybersecurity. Rolf Stadler has been Editor-in-Chief of IEEE Transactions on Network and Service Management (TNSM) 2014-2017.
\end{IEEEbiography}
\end{document}